%% file: paperArxiv2.tex
\documentclass[11pt]{article}
\usepackage{microtype}
\usepackage{fullpage}

\usepackage{amsmath}
\usepackage{amssymb}
\usepackage{amsthm}
\usepackage{comment}
\usepackage[sc]{mathpazo} 
\usepackage{paralist}

\newtheorem{theorem}{Theorem}
\newtheorem{lemma}[theorem]{Lemma}
\newtheorem{definition}[theorem]{Definition}

\usepackage{graphicx}
\usepackage[usenames,dvipsnames]{color}
\usepackage{rotating}
\usepackage{algpseudocode}
\usepackage{algorithm}
\usepackage{url}

\allowdisplaybreaks

\newtheorem{obs}[theorem]{Observation}

\newtheorem{cor}[theorem]{Corollary}
\newtheorem{fact}[theorem]{Fact}

\newcommand{\superscript}[1]{\ensuremath{^{\textrm{#1}}}}

\newcommand{\johnc}[1]{{\sc \textcolor{red}{(xiii) #1}}}
\newcommand{\johnd}[1]{{\sc \textcolor{red}{(xiv) #1}}}
\renewcommand{\johnc}[1]{}
\renewcommand{\johnd}[1]{}

\newcommand{\LL}[1]{{\textcolor{Blue}{#1}}}

\newcommand{\checked}{}

\newwrite\notationfile
\openout\notationfile=\jobname.notation

\newcommand{\jnc}[3]{\write\notationfile{$\mathrm{\string\string\string#1}$ & \ensuremath{#2} & #3 \string\\}\newcommand{#1}{\ensuremath{#2}}}

\newcommand{\opIns}{\ensuremath{\text{\sc Insert}}}
\newcommand{\opEm}{\ensuremath{\text{\sc Extract-Min}}}
\newcommand{\opDc}{\ensuremath{\text{\sc Decrease-Key}}}
\newcommand{\opMg}{\ensuremath{\text{\sc Meld}}}

\newcommand{\eIns}{\ensuremath{\text{\sc Evolve-Insert}}}
\newcommand{\eDc}{\ensuremath{\text{\sc Evolve-Decrease-Key}}}
\newcommand{\eDmr}{\ensuremath{\text{\sc Evolve-Designated-Minimum-Root}}}
\newcommand{\eEm}{\ensuremath{\text{\sc Evolve-Extract-Min}}}
\newcommand{\eBs}{\ensuremath{\text{\sc Evolve-Big-Small}}}
\newcommand{\ePerm}{\ensuremath{\text{\sc Evolve-Permute}}}

\newcommand{\tent}[1]{\overline{#1}}

\newcommand{\subop}[2]{{\sc #1(\ensuremath{#2})}}

\newcommand{\rcfvalue}[2]{\frac{2^{#1-1}}{\log #2}}
\newcommand{\rcgvalue}[1]{\frac{1}{2 \log #1}}
\newcommand{\cst}{\frac{n  }{\rcf {} \log n}} 
\newcommand{\pp}{++}

\newcommand{\consd}{4}

\jnc{\Ids}{\ensuremath{\Xi}}{?}
\jnc{\Distinct}{\ensuremath{\xi}}{?}
\jnc{\opa}{\ensuremath{a}}{Description}
\jnc{\opb}{\ensuremath{b}}{Description}
\jnc{\opc}{\ensuremath{c}}{Description}
\jnc{\opA}{\ensuremath{A}}{Description}
\jnc{\opB}{\ensuremath{B}}{Description}
\jnc{\opC}{\ensuremath{C}}{Description}
\jnc{\subops}{\eta}{Description}
\jnc{\pointers}{\rho}{Description}
\jnc{\pointer}{p}{Description}
\jnc{\state}{s}{Description}
\jnc{\rank}{r}{Description}
\jnc{\alg}{\mathcal{A}}{The algorithm with $o(\log \log n)$-time decrease key which is being assumed to exist.}
\jnc{\cem}{c}{The (maximum of) the constants hidden in the big-O of the $O(\log n)$-time \opEm\ and \opIns\ operations. }
\jnc{\rce}{j(n)}{Description}
\jnc{\rcf}{f(n)}{Description}
\jnc{\rcg}{g(n)}{Description}
\jnc{\rch}{h}{Maximum fraction of the time spent doing fix-up}
\jnc{\rcl}{\ell}{Amortized runtime is $\ell \log n$ per round}
\jnc{\fdc}{dc(n)}{Cost of decrease-key}
\jnc{\vio}{v}{Description}
\jnc{\akey}{x}{A key value or node}
\jnc{\theset}{\mathcal{S}}{Set of items stored by the data structure}
\jnc{\acem}{\ensuremath{e}}{...}
\jnc{\dcc}{dc}{Number of \opDc\ performed in a permutation evolution in round $i$}
\jnc{\vs}{v}{Size of the violation sequence of the \opEm\ in round $i$}
\jnc{\csmall}{z}{At least $z$ of the rounds are small}
\jnc{\rootsize}{m}{The size of a particular root}
\jnc{\noninc}{N}{Nonincremental sibling set}
\jnc{\bigroot}{r}{Location of largest root}
\jnc{\rcd}{{\ensuremath{j(n)}}}{Used in rank definition}
\jnc{\round}{\circ}{Round}
\jnc{\umc}{k}{Number of unmarked children}
\jnc{\Vio}{V}{Violation set: those nodes that become marked in a extract-min evolution}
\jnc{\Classification}{C}{Classification in a big-small evolution}
\jnc{\Evo}{\Psi}{A sequence of evolutions}
\jnc{\evo}{\psi}{One evolution in a sequence of evolutions}
\jnc{\runb}{R(B)}{Runtime to execute operation sequence $B$}
\jnc{\rcm}{m}{Constant involved in size of number of decrease-key operations in small round}
\closeout\notationfile

\newenvironment{description*}%
  {\vspace{-1ex}\begin{description}%
    \setlength{\itemsep}{-0.5ex}%
    \setlength{\parsep}{0pt}}%
  {\end{description}}
\newenvironment{itemize*}%
  {\vspace{-1ex}\begin{itemize}%
    \setlength{\itemsep}{-0.5ex}%
    \setlength{\parsep}{0pt}}%
  {\end{itemize}}
\newenvironment{enumerate*}%
  {\vspace{-1ex}\begin{enumerate}%
    \setlength{\itemsep}{-0.5ex}%
    \setlength{\parsep}{0pt}}%
  {\end{enumerate}}

\usepackage{titlesec}

\titlespacing{\section}{0pt}{3pt}{12pt}
\titlespacing{\subsection}{0pt}{3pt}{12pt}
\titlespacing{\subsubsection}{0pt}{3pt}{12pt}
\titlespacing{\paragraph}{0pt}{0pt}{12pt}
\titlespacing{\subparagraph}{0pt}{3pt}{12pt}

\titleformat{\section}[runin] 
{\normalfont\bfseries} 
{\thesection}{1em}{}[.] 

\titleformat{\subsection}[runin] 
{\normalfont\bfseries} 
{\thesubsection}{1em}{}[.] 

\titleformat{\subsubsection}[runin] 
{\normalfont\bfseries} 
{\thesubsubsection}{1em}{}[.] 

\usepackage{enumitem}
\setlist{noitemsep,topsep=0pt}
\setitemize{noitemsep,topsep=0pt}
\setenumerate{noitemsep,topsep=0pt}

\title{Why some heaps support\\ constant-amortized-time decrease-key operations,\\ and others do not}

\author{John Iacono\thanks{Research supported by NSF Grant CCF-1018370.}}

\date{{\sc The Polytechnic Institute of New York University}\\Brooklyn, New York, USA}

\begin{document}

\maketitle

\includecomment{shortonly}
\excludecomment{fullonly}

\newcommand{\jlabel}[1]{\label{short:#1}}
\newcommand{\jref}[1]{\ref{short:#1}}
\newcommand{\jeqref}[1]{\eqref{short:#1}}
\newcommand{\shortfull}[2]{#1}

\newcommand{\proofapp}[1]{For proof, see Lemma~\ref{full:#1} in the Appendix.}

\includecomment{fullonly}
\excludecomment{shortonly}

\renewcommand{\jlabel}[1]{\label{full:#1}}
\renewcommand{\jref}[1]{\ref{full:#1}}
\renewcommand{\jeqref}[1]{\eqref{full:#1}}
\renewcommand{\shortfull}[2]{#2}

\renewcommand{\proofapp}[1]{}

\titlespacing{\section}{0pt}{3pt}{12pt}
\titlespacing{\subsection}{0pt}{3pt}{12pt}
\titlespacing{\subsubsection}{0pt}{3pt}{12pt}
\titlespacing{\paragraph}{0pt}{0pt}{12pt}
\titlespacing{\subparagraph}{0pt}{3pt}{12pt}

\titleformat{\section}
{\LARGE\bfseries} 
{\thesection}{1em}{} 

\titleformat{\subsection}
{\Large\bfseries} 
{\thesubsection}{1em}{}

\titleformat{\subsubsection}
{\large\bfseries} 
{\thesubsubsection}{1em}{}



%
%


\input{paperArxiv2main}


\label{refs}

\bibliographystyle{alpha}
\bibliography{dblp,bib}

\end{document}

%% file: paperArxiv2main.tex

\begin{abstract}
A lower bound is presented which shows that a class of heap algorithms in the pointer model with only heap pointers must spend $\Omega \left( \frac{\log \log n}{\log \log \log n} \right)$ amortized time on the \opDc\ operation (given $O(\log n)$ amortized-time \opEm). 
Intuitively, this bound shows the key to having $O(1)$-time \opDc\ is the ability to sort $O(\log n)$ items in $O(\log n)$ time; Fibonacci heaps [M.~.L.~Fredman and R.~E.~Tarjan. J.~ACM 34(3):596-615 (1987)] do this through the use of bucket sort. Our lower bound also holds no matter how much data is augmented; this is in contrast to the lower bound of Fredman [J. ACM 46(4):473-501 (1999)] who showed a tradeoff between the number of augmented bits and the amortized cost of \opDc. A new heap data structure, the \emph{sort heap}, is presented. This heap is a simplification of the heap of Elmasry [SODA 2009: 471-476] and shares with it a $O(\log \log n)$ amortized-time \opDc, but with a straightforward implementation such that our lower bound holds. Thus a natural model is presented for a pointer-based heap such that the amortized runtime of a self-adjusting structure and amortized lower asymptotic bounds for \opDc\ differ by but a $O(\log \log \log n)$ factor.
\end{abstract}
\vfill

\begin{shortonly}
\begin{center}
We have chosen to include a full version which contains all proofs, in addition to some results omitted from the 10-page version and a more detailed and generously typeset presentation. The full version appears starting on page~\pageref{full}. The references are at the end, starting on page~\pageref{refs}.
\end{center}
\vfill
\end{shortonly}

\pagebreak

\begin{fullonly}
\begin{sidewaystable} \footnotesize 
\begin{tabular}{c|ccccc}
&Fibonacci Heap& Pointer\\
&Rank-Pairing Heap &  Rank-Pairing & Pairing heap & Elmasry's heap & Sort heap\\ \hline \hline
Introduced &  \parbox{0.5in}{\center \cite{DBLP:journals/jacm/FredmanT87} \cite{DBLP:journals/siamcomp/HaeuplerST11}} & \parbox{1.5in}{\center Variant of Rank-Pairing heap described here} & \cite{DBLP:journals/algorithmica/FredmanSST86}&\cite{DBLP:conf/esa/Elmasry10,DBLP:conf/soda/Elmasry09}  & This paper\\ 
\hline
Upper bound for round & $O(\log n)$ & $O(\log n \log \log n)$ & $O(\log n2^{\sqrt{\log \log n}})$ & $O(\log n \log \log n)$ & $O(\log n \log \log n)$\\
&  & & \cite{DBLP:conf/focs/Pettie05} \\ 
\hline
Trivial lower bound & $\Omega(\log n)$ & $\Omega(\log n)$ & $\Omega(\log n)$ & $\Omega(\log n)$ & $\Omega(\log n)$\\
\parbox{1.5in}{\center Fredman's lower bound \cite{DBLP:journals/jacm/Fredman99}} & $\Omega(\log n)$ &$\Omega(\log n)$ &$\Omega(\log n \log \log n)$  &Does not apply &Does not apply\\
Our lower bound& Does not apply &  $\Omega\left( \frac{\log n \log \log n}{\log \log \log n} \right)$ &  $\Omega\left( \frac{\log n \log \log n}{\log \log \log n} \right)$ & Does not apply &  $\Omega\left( \frac{\log n \log \log n}{\log \log \log n} \right)$\\
\hline
Tightness of analysis & $O(1)$ & $O(\log \log \log n)$ & $O(2^{\sqrt{\log \log n}} \log \log  n)$ & $O(\log \log n)$ & $O(\log \log \log n)$\\
\hline
\hline
Augmented data & Yes & Yes & No & No & No
\end{tabular}
\caption{Comparison of various heaps, giving their amortized runtimes per round, where a round consists of one \opIns, $\log n$ \opDc\ operations, and one \opEm\ on a heap of size $n$. Tightness of analysis is the ratio of the best upper bound to the best lower bound. }
\jlabel{t1}
\end{sidewaystable}
\end{fullonly}


\section{Introduction}

While \opIns\ and \opEm\ are supported by all priority queues, there is one additional operation which is of use in some algorithms: \opDc. The fast execution of \opDc\ is vital to the runtime of several algorithms, most notably Dijkstra's algorithm \cite{springerlink:10.1007/BF01386390} for single-source shortest paths in a graph. The constant-amortized-time implementation of the \opDc\ operation is the defining feature of the \emph{Fibonacci heap} \cite{DBLP:journals/jacm/FredmanT87} data structure. 

In \cite{DBLP:journals/algorithmica/FredmanSST86}, a new heap called the \emph{pairing heap} was introduced. The pairing heap is a self-adjusting heap, whose design and basic analysis closely follows that of splay trees \cite{DBLP:journals/jacm/SleatorT85}. They are much simpler in design than Fibonacci heaps, and they perform well in practice.
\begin{fullonly}
 They were included in the (pre-STL) C++ data structures library. We have shown that they have a working-set like runtime bound \cite{DBLP:conf/swat/Iacono00}. 
\end{fullonly}

It was conjectured at the time of their original presentation that pairing heaps had the same $O(\log n)$ amortized time \opEm, and $O(1)$ amortized time \opIns\ and \opDc\ as Fibonacci heaps, however, at their inception only a $O(\log n)$ amortized bound was proven for all three operations. Stasko and Vitter~\cite{DBLP:journals/cacm/StaskoV87} provided some simulation results which showed that $O(1)$ \opDc\ appeared likely. We have shown that \opIns\ does in fact have $O(1)$ amortized time \cite{DBLP:journals/corr/abs-1110-4428,DBLP:conf/swat/Iacono00}. 

However in \cite{DBLP:journals/jacm/Fredman99}, Fredman refuted the conjectured constant-amortized-time \opDc\ in pairing heaps by proving that pairing heaps have a lower bound of $\Omega(\log \log n)$ on the \opDc\ operation. The result he proved was actually more general: he created a model of heaps which includes both pairing heaps and Fibonacci heaps, and produced a tradeoff between the number of bits of data augmented and a lower bound on the runtime of \opDc. Pairing heaps have no augmented bits, and were shown to have a $\Omega(\log \log n)$ amortized lower bound on \opDc\ while he showed that in his model a $O(1)$ \opDc\ requires $\Omega(\log \log n)$ bits of augmented information per node, which is the number of bits of augmented information used by Fibonacci heaps and variants%
\begin{fullonly}
 (A number called rank, which is an integer with logarithmic range is stored in each node)%
\end{fullonly}
 . More recently, Pettie has shown a upper bound of $2^{O(\sqrt{\log \log n})}$  for \opDc\ in pairing heaps \cite{DBLP:conf/focs/Pettie05}. It remains open where in the range $\Omega(\log \log n) \ldots 2^{O(\sqrt{\log \log n})}$ the true cost of \opDc\ in a pairing heap lies.

Elmasry has shown recently that a simple variant of pairing heaps has $O(\log \log n)$ amortized time \opDc\ operation \cite{DBLP:conf/esa/Elmasry10,DBLP:conf/soda/Elmasry09}. However, 
\begin{fullonly}
for technical reasons described later, including a non-standard implementation of \opDc, 
\end{fullonly}
this variant is not in Fredman's model and thus the $\Omega(\log \log n)$ lower bound does not apply. 

So, it would seem from the preceding exposition that the situation has been essentially resolved: Fibonacci heaps are complex but optimal, while the elegant pairing heaps (and Elmasry's variant) are as good as a self-adjusting structure can get (One defining feature of a \emph{self-adjusting} structure is that they store no augmented data in every node). In the case of dictionaries, we have the hope of instance-based optimality, as evidenced by the dynamic optimality conjecture \cite{DBLP:journals/jacm/SleatorT85}, while in the case of heaps, self-adjusting structures can not even achieve optimal amortized asymptotic runtimes. We will now propose an alternate interpretation of the facts which leads to a much nicer conclusion.

Fredman's model is nuanced, and has limitations which cause us to introduce here a new model, which we call the \emph{pure heap} model. We first informally describe our model, and then describe how it differs from Fredman's.
A pure heap is a pointer-based forest of rooted trees, each node holding one key, that obeys the heap property (the key of the source of every pointer is smaller than the key of the destination); the nodes may be augmented, and all operations must be valid in the pointer model; heap pointers are only removed if one of the nodes is removed or if a \opDc\ is performed on the node that a heap pointer points to. 
This model is both simple and captures the spirit of many heaps  like the pairing heap, and is meant to be a clean definition analogous to that of the well-established binary search tree (BST) model \cite{DBLP:journals/siamcomp/Wilber89}.
However, Fibonacci heaps are not pure model heaps in the standard textbook presentation \cite{clrs} for the following reason: each node is augmented with a $\log \log n$-bit integer (in the range $1\ldots \log n$), and in the implementation of \opEm\ it is required to separate a number of nodes, call it $k$, into groups of nodes with like number. This is done classically using bucket sort in time $O(k+\log n)$. However, bucket sort using indirect addressing to access each bucket is distinctly not allowed in the pointer model; realizing this, there is a note in the original paper \cite{DBLP:journals/jacm/FredmanT87}  showing how by adding another pointer from every node to a node representing its bucket, bucket sort can be be simulated at no additional asymptotic cost. However, these pointers are non-heap pointers to nodes with huge indegree which store no keys and thus violate the tree requirement of the pure heap model, as well as the spirit of what we usually think of as a heap. 

There have been several alternatives to Fibonacci heaps presented with the same amortized runtimes which are claimed to be simpler than the original. These include
thin heaps \cite{DBLP:journals/talg/KaplanT08},
violation heaps \cite{DBLP:journals/dmaa/Elmasry10},
and rank-pairing heaps \cite{DBLP:journals/siamcomp/HaeuplerST11}.
All of these heaps also use indirect addressing or non-heap pointers. The rank-pairing heap has the cleanest implementation of all of them, and implements \opDc\ by simply disconnecting the node from it parents and not employing anything more complicated like the cascading cuts of Fibonacci heaps.

We can easily modify Fibonacci heaps and the aforementioned alternatives to only use heap pointers by simply 
using $O(\log \log n)$ time to determine which of the $O(\log n)$ buckets each of the $k$ keys are in. We call such a variant of rank pairing heaps a \emph{pointer Rank-Pairing} heap\shortfull{.}{, and list in in Table~\jref{t1}.} However, this alteration has the effect of increasing the time of \opDc\ in a Fibonacci heaps and their variants to $O(\log \log n)$. 

Fredman's model differs from ours in several regards. First, Fredman's model requires that comparisons can not be performed unless the nodes being compared must be linked by a heap pointer after the comparison. Second, in the course of an \opEm, any two children of the former root can be randomly accessed and compared at unit cost. Third, the number of augmented bits per node is a parameter of the model. The first restriction, while it admits pairing heaps and Fibonacci heaps, excludes Elmasry's variant. This is because Elmasry's variant sorts the keys in nodes in order to determine how to link them; such sorting is directly against what is allowed in Fredman's model. Our sort heaps, presented in section~\ref{full:s:sortheap}\shortfull{ in the appendix}{}, are not in Fredman's model for the same reason. Also, subjectively, we find the first restriction a bit odd, but it appears to be vital to the result. The second difference means that a fundamental cost in the pointer model is not counted: moving pointers to reach the desired nodes to compare or otherwise manipulate. So, compared to our model, Fredman's model is more restrictive because of the first condition, and more permissive with the second condition. We feel our model is more natural.
\begin{shortonly}
A more detailed discussion of how our model differs from Fredman's can be found in \S\ref{full:s:fh} in the appendix.
\end{shortonly}

Thus, we conclude that the reason that Fibonacci heaps have fast \opDc\ is not (only) because of the augmented bits as Fredman's result suggests, but rather because they depart from the pure heap model. We prove that any pure-heap-model heap has an $\Omega\left( \frac{\log \log n}{\log \log \log n} \right)$ amortized lower bound on decrease key, \emph{no matter how many bits of data each node is augmented with}. 
In our view, Fibonacci heaps are a typical RAM-model structure that squeeze out a $\log \log n$ factor over the best structure in a natural pointer-based model by beating the sorting bound using bit tricks (of which bucket sort is a very primitive example).

Given our lower bound, we still have the issue that it does not apply to any known self-adjusting heaps known to have fast \opDc\ times. We rectify this in \shortfull{the appendix in }{}\S\ref{full:s:sortheap}, by introducing the \emph{sort heap}. This heap is simply Elmasry's heap with the non-standard \opDc\ replaced with the standard one, or to put in another way, it is identical to the pairing heap except for the choice of pairings used to implement the \opEm\ operation. Our sort heap has an $O(\log \log n)$ amortized \opDc\ operation, and features an analysis that differers markedly from other self-adjusting heaps with fast \opDc.

The paper is structured as follows: 
\shortfull%
{a review of the priority queue ADT is in the appendix in \S\ref{full:s:pq},}%
{we begin by reviewing the priority queue ADT in \S\ref{full:s:pq}, }%
 and then formally present our pure heap model (\S\jref{s:ph})\shortfull{.}{ followed by a comparison with Fredman's generalized pairing heap model (\S\jref{s:fh}).}
\shortfull%
{The main result, our lower bound, is presented in \S\jref{s:lb}. The presentation of sort heaps is in the appendix in \S\ref{full:s:sortheap}. }%
{The main result, our lower bound, is presented in \S\jref{s:lb}, and this is followed by sort heaps which are presented in \S\jref{s:sortheap}. }%
In \S\jref{s:comments} \shortfull{closing comments appear.}{with some thoughts and directions for further work.}

\begin{fullonly}
\section{The priority queue abstract data type} \jlabel{s:pq}

A \emph{priority queue} supports the following operations to maintain a totally ordered set \theset:

\begin{itemize}

\item $\pointer= $\opIns$(\akey)$: Inserts the key $\akey$ into $\theset$ and returns a pointer $p$ used to perform future \opDc\ operations on this key.

\item $\akey= $\opEm$()$: Removes the minimum item in $\theset$ from $\theset$ and returns it.

\item \opDc$(\pointer, \Delta \akey)$: Reduces the key value pointed to by $\pointer$ by some non-negative amount $\Delta \akey$.

\end{itemize}

The key values can be in any form so long as they are totally ordered and an $O(1)$-time comparison function is provided. (This is more permissive than saying they they are comparison-based. We do not restrict algorithms from doing things like making decisions based on individual bits of a key). Some priority queues, including Fibonacci heaps and our sort heaps, also efficiently support the \opMg\ operation where two priority queues are combined into one.
\end{fullonly}

\section{The pure heap model} \jlabel{s:ph}

Here we define the \emph{pure heap} model, and how priority queue operations on data structures in this model are executed in it.

\subsection{Structural invariants and terminology}

The pure heap model requires that at the end of every operation, the data structure is an ordered forest of general heaps. 
Each node is associated with a key $\akey \in \theset$. We will use $x$ both to refer to a key and the node in the heap containing the key. Inside each node is stored the key value, pointers to the parent, leftmost child and right sibling of the node, along with other possible augmented information.

The \emph{structure} of the heap is the shape of the forest, without regard to the contents of the nodes. The \emph{location} of a node is its position in the forest of heaps relative to the right (e.g. a node could be described as being the fifth child from the right of a node which is the third child from the right of the fourth root from the right. The idea is that location is invariant under adding new siblings to the left).

An algorithm in the pure heap model implements the priority queue operations as follows:
\shortfull{}{\begin{itemize\shortfull{*}{}}}
\shortfull{}\item \opIns\ operations are executed\shortfull{ at unit cost}{} by adding the new item as a new leftmost heap. \shortfull{}{The cost is defined to be 1.}
\shortfull{}\item \opDc\ is executed\shortfull{ at unit cost}{}  by disconnecting from its parent the node containing the key to be decreased (if it is not a root), decreasing the key, and then placing it as the leftmost root in the forest of heaps. \shortfull{}{The cost is defined to be 1.}
\shortfull{}\item A \opEm\ operation is performed by first executing a sequence of pointer-based suboperations which are fully described below in \S\jref{suboperations}. After executing the suboperations, the forest is required to be monoarboral (i.e.~have only one heap). 
Thus, the root of this single tree has as its key the minimum key in $\theset$. This node is then removed, the key value is returned, and its children become the new roots of the forest. The cost is the number of suboperations performed. 
\shortfull{}{\end{itemize\shortfull{*}{}}}

Note that some data structures are not presented exactly in the framework as described above but can be easily put into this mode by being lazy. For example, in a pairing heap, the normal presentations of \opIns\ and \opDc\ cause an immediate pairing with the single existing root. However, such pairings can easily be deferred until the next \opEm, thus putting the resulting structure in our pure-heap framework. 
\begin{fullonly}
For a more elaborate example of transforming a heap into one based on pairings, see \cite{DBLP:conf/wae/Fredman99}.
\end{fullonly}

\subsection{Executing a \opEm; suboperations} \jlabel{suboperations}

\shortfull{To}{Obviously, in order to} execute an \opEm, the minimum must be determined. 
\shortfull{In an \opEm, the forest of heaps must be combined into a single heap using an operation called \emph{pairing}, which takes two roots and attaches the root with larger key value as the leftmost child of the root with smaller key value.}%
{At the beginning of the \opEm\ operation, the structure of the heap may consist of a forest of many heap-ordered trees, the pure-heap model requires that these trees be combined into one tree through a process called \emph{pairing}. The pairing operation takes two roots and attaches the root with larger key value as the leftmost child of the root with smaller key value.}%
\begin{shortonly}
\footnote{(Note that while the pairing operation brings to mind pairing heaps, it is the fundamental building block of many heaps. Even the \emph{skew heap} \cite{DBLP:journals/siamcomp/SleatorT86}, which seems at first glance to not use anything that looks like the pairing operation, can be shown in all instances to be able to be transformed into a pairing-based structure at no decrease in cost \cite{DBLP:conf/wae/Fredman99}).}
\end{shortonly}
\begin{fullonly}
(Note that while the pairing operation brings to mind pairing heaps, it is the fundamental building block of many heaps. Even the \emph{skew heap} \cite{DBLP:journals/siamcomp/SleatorT86}, which seems at first glance to not use anything that looks like the pairing operation, can be shown in all instances to be able to be transformed into a pairing-based structure at no decrease in cost \cite{DBLP:conf/wae/Fredman99}).  We require this process to happen in the pointer model, where there is some constant number of pointers that start at the leftmost root and move around and perform pairings. 

For the purposes of the analysis, it is needed to have a very fine view of what constitutes a constant-time suboperation, but, what is presented below is constant-time-equivalent to other natural ways of having a pointer model view with the basic primitive being pairing of the roots.
\end{fullonly}

In the execution of the extract-min operation, the use of some constant number $\pointers$ of pointers $\pointer_1, \pointer_2, \ldots \pointer_\rho$ is allowed. They are all initially set to the leftmost root. The constant $\pointers$ is a parameter of the model. \shortfull{}{These are the suboperations that are allowed to execute the \opEm\ operation:}

\begin{shortonly}
The suboperations include: move one of the pointers to the parent, left or right sibling, or leftmost child; pair the nodes pointed to by two pointers; and copy pointers.
A special suboperation, \subop{End}{}, signifies the execution of the \opEm\ is complete.
A full list of suboperations and their precise definitions and preconditions can be found in the \S\ref{full:suboperations} in the appendix.
\end{shortonly}
\begin{fullonly}
\begin{enumerate}

\item \jlabel{so:chpar} \subop{HasParent}{i}: Return if the node pointed to by $\pointer_i$ has a parent.

\item \subop{HasLeftSibling}{i}: Return if the node pointed to by $\pointer_i$ has a left sibling.

\item \subop{HasRightSibling}{i}: Return if the node pointed to by $\pointer_i$ has a right sibling.

\item \subop{HasChildren}{i}: \jlabel{so:chchild} Return if the node pointed to by $\pointer_i$ has any children.

\item \subop{Compare}{i,j}:\jlabel{so:compare} Return if the key value in the node pointed to by $p_i$ is less than or equal to the key value in the node pointed to by $p_j$.

\item \subop{Pair}{i,j}:\jlabel{so:pair} Perform a \emph{pairing} on two pointers $\pointer_i$ and $\pointer_j$ where the tree that $p_j$ points to is attached to $\pointer_i$ as its leftmost subtree. It is a required precondition of this suboperation that both $p_i$ and $p_j$ point to roots, and that this was verified by the \textsc{HasParent} suboperation, and that the key value in the node pointed to to $p_i$ is smaller than the key value in the node pointed to by $p_j$, and that is was verified by the \textsc{Compare} suboperation.

\item \subop{Set}{i,j}: Set pointer $\pointer_i$ to point to the same node as $\pointer_j$.

\item \subop{MoveToParent}{i}: Move a pointer $\pointer_i$ to the parent of the node currently pointed to.
It is a precondition of this operation that the node that $\pointer_i$ points to has a parent, and that this was verified by the \textsc{HasParent} operation.

\item \subop{MoveToLeftmostChild}{i}: Move the pointer $\pointer_i$ to the leftmost child.
It is a precondition of this operation that the node that $\pointer_i$ points to has children, and that this was verified by the \textsc{HasChildren} operation.

\item \subop{MoveToRightSibling}{i}: Move the pointer $\pointer_i$ to the sibling to the right.
It is a precondition of this operation that the node that $\pointer_i$ points to has a right sibling, and that this was verified by the \textsc{HasRightSibling} operation.

\item \subop{MoveToLeftSibling}{i}: \jlabel{so:mvleft} Move a pointer $\pointer_i$ to the sibling to the left
It is a required precondition of this operation that the node that $\pointer_i$ points to has a left sibling, and that this was verified by the \textsc{HasLeftSibling} operation.

\item \subop{End}{}: \jlabel{so:end} Marks the end of the suboperation sequence for a particular \opEm.
It is a required precondition of this operation that the forest of heaps contains only one heap, and that this was verified through the use of the \subop{HasParent}{i}, \subop{HasLeftSibling}{i}, and \subop{HasRightSibling}{i} on a pointer $p_i$ that points to the unique root.

\end{enumerate}

Operations \jref{so:chpar}-\jref{so:chchild} return a boolean; the remainder have no return value.

\end{fullonly}
The total number of suboperations, including the parameters, is defined to be $\subops$. Observe that $\subops=\Theta(\pointers^2)$, which is $\Theta(1)$ since $\rho$ is a constant. 
A sequence of suboperations is a \emph{valid} implementation of the \opEm\ operation if all the preconditions of each suboperation are met and the last suboperation is an \subop{end}{}.

In the pure heap model, the only thing that differentiates between different algorithms is in the choice of the suboperations to execute \opEm\ operations. In these operations it is the role of the particular \emph{heap algorithm} to specify which suboperations should be performed for each \opEm. We place no restrictions as to how an algorithm determines the suboperation sequence for each \opEm\ other than the suboperation sequence must be valid.

This definition of an algorithm encompasses and is more permissive than allowing the algorithm to make decisions to be made based on some data augmented at any node. This is because augmented information is just one type of function of the previous operations whereas we allow the algorithm to decide what suboperations to perform in any manner, subject only to determinism. We also note that while the definition of the pairing operation enforces the heap structure, where every parent is smaller than its children, the algorithm is not restricted from looking at, for example, the individual bits of a key and deciding which pairings to perform based on this. 

\begin{fullonly}
To summarize, our notion of an algorithm allows the algorithm, at every step of determining which suboperation to execute next, to make decisions in any deterministic manner. These computations to determine which suboperation to execute next have no cost; only the execution of the suboperations themselves incurs a cost.
\end{fullonly}

\begin{fullonly}
\section{Fredman's model} \jlabel{s:fh}

Fredman's model, which he calls \emph{generalized pairing heaps} differ from our \emph{pure heap} model in a number of significant aspects:

\shortfull{}{\begin{itemize\shortfull{*}{}}}
\shortfull{}{\item} Generalized pairing heaps are parametrized by the number of bits of augmented data allowed at each node. Pure heaps allow an algorithm to branch as any function of the past; this is equivalent to allowing unlimited augmented data.

\shortfull{}{\item} Generalized pairing heaps limit how performing a comparison can be done by the algorithm. In the process of executing an \opEm\ operation, suppose a comparison is performed between two nodes, neither of which will be removed by the operation. This limits the number of comparisons to be performed in a \opEm\ to be linear in the number of pairings performed. In a generalized pairing heap, the result of this comparison can not be used to determine the future actions of the algorithm. In pure heaps, the result of such a comparison can be used. It is this crucial difference that places Elmasry's variant of the pairing heap and our sort heap in the pure heap model, but not in the generalized pairing heap model. These algorithms perform \opDc\ operations on selected roots after performing comparisons to sort them; sorting roots is out-of-model in the generalized pairing heap model; this is easy to see because sorting requires doing a super-linear number of comparisons. 

\shortfull{}{\item} Generalized pairing heaps do not take into account the time needed to access the items that are to be paired; arbitrary pairing of roots is allowed at unit cost. In the pure heap model, one needs to move pointers to the nodes to be paired using pointer-model-operations on the heap, and these operations must be paid for.
\shortfull{}{\end{itemize\shortfull{*}{}}}
\end{fullonly}

\section{Lower Bound} \jlabel{s:lb}

\subsection{Statement of result} \jlabel{sec:sor}

\begin{theorem} \jlabel{th:main}
In the pure heap model with a constant number $\pointers$ pointers, if\/ \opEm\ and \opIns\ have an amortized cost of $O(\log n)$, then \opDc\ has an amortized cost of $\Omega \left( \frac{\log \log n}{\log \log \log n} \right)$.\end{theorem}

The proof will follow by contradiction and will consume the rest of this section. Assume that there is a pure heap model algorithm $\alg$ where \opEm\ and \opIns\ have an amortized cost of at most $\cem \log n$, for some constant $\cem$, and \opDc\ has an amortized cost of at most  $\fdc$, for some $\fdc=o\left( \frac{\log \log n}{\log \log \log n} \right)$; since this is a lower bound we can also safely assume that $\fdc=\omega(1)$. 
The existence of the algorithm  $\alg$,  the constant $\cem$ and the function $\fdc$ will be assumed in the definitions and lemmas that follow. A sufficiently large $n$ is also assumed, as there are several places noted in what follows, such as using the results of asymptotic expressions, where this is required.

\subsection{Overview of proof}

The proof is at its core an adversary argument. Such arguments look at what the algorithm has done and then decide what operations to do next in order to guarantee a high runtime. But, our argument is not straightforward as it works on sets of sequences of operations rather than a single operation sequence. There is a hierarchy of things we manipulate in our argument:

\begin{fullonly}
\begin{description}
\end{fullonly}

\shortfull{\emph{Suboperaton.}}{\item[Suboperaton.]} The suboperations of \S\jref{suboperations} are the very basic unit-cost primitives that can be used to implement \opEm, the only operation that does not have constant actual cost. It is at this level that definitions have been made to enforce pointer model limitations.
\shortfull{\emph{Operation.}}{\item[Operation.]} We use \emph{operation} to refer to a priority queue operation. \shortfull{}{In this proof, the adversary will only use \opIns, \opDc, and \opEm.}
\shortfull{\emph{Sequence.}}{\item[Sequence.]} Operations are combined to form sequences of operations. 
\shortfull{\emph{Set of operation sequences.}}{\item[Set of operation sequences.]} Our adversary does not just work with a single operation sequence but rather with sets of operation sequences. These sets are defined to have certain invariants on the heaps that result from running the sequences that bound the size of the sets of operation sequences under consideration.
\shortfull{\emph{Evolutions.}}{\item[Evolutions.]} We use the word \emph{evolution} to refer to a function the adversary uses to take a set of operation sequences, and modify it. The modifications performed are to append operations to sequences, remove sequences, and to create more sequences by taking a single sequence and appending different operations to the end.
\shortfull{\emph{Rounds.}}{\item[Rounds.]} Our evolutions are structured into \emph{rounds}.

\begin{fullonly}
\end{description}
\end{fullonly}

The proof will start with a set containing a single operation sequence, and then perform rounds of evolutions on this set; the exact choice of evolutions to perform will depend on how the algorithm executes the sequences of operations in the set. The evolutions in a round are structured in such a way that most rounds increase the size of the set of operations. After sufficiently many rounds, an upper bound on the maximum size of the set of operation sequences will be exceeded, thus giving a contradiction.

Our presentation is structured as follows: 
In \S\jref{s:rank}, we define a rank function.
In \S\jref{s:monotonic}-\jref{s:augmented} we give some invariants and facts about the sequences of operations we will be considering.
In \S\jref{s:evolutiondefinition} we introduce the idea of a set of sequences and explain the invariants of the sets that will be maintained.
We introduce the idea of an evolution in \S\jref{s:evolving} and then describe several types of evolutions in \S\jref{s:einsert}-\jref{ev:perm}. 
These evolutions are structured into rounds in \S\jref{s:rounds}, technical lemmas about rounds appear in \S\jref{s:ubot}, and the final work to obtain the contradiction is in \S\jref{s:pit}.

\subsection{Ranks: definitions and useful facts} \jlabel{s:rank}

\subsubsection{Motivation}

As in many previous works on heaps and trees, the notion of the \emph{rank} of a node in the heap is vital. The rank of a node is meant to be a rough proxy for the logarithm of the size of the subtree of the node. While the basic analysis of pairing heaps and splay trees \cite{DBLP:journals/algorithmica/FredmanSST86,DBLP:journals/jacm/SleatorT85} use exactly this as the rank, the definition of rank here is more delicate and is an extension of that used in \cite{DBLP:journals/jacm/Fredman99}. As in \cite{DBLP:journals/jacm/Fredman99}, rank here is always a nonnegative integer. In previous definitions of rank, the value typically depended only on the current structure of the heap (One exception to this has been in order to get better bounds on \opIns, nodes are treated differently for potential purposes depending on whether or not they will ever be deleted. See \cite{DBLP:journals/cacm/StaskoV87,DBLP:journals/corr/abs-1110-4428,DBLP:conf/swat/Iacono00}). Here, however, the definition is more nuanced in that for the purposes of the analysis only, nodes are classified into \emph{marked} and \emph{unmarked} categories based on the history of the structure, and these marks, along with the current structure of the heap, are used to compute the rank of each node.

\subsubsection{Definition} \jlabel{sec:def}

For ease of presentation, the rank of a node is defined in terms of the function $\rce=2\fdc+1$.
\begin{fullonly}

\end{fullonly}
The general idea is to have the rank of a node be maintained so it is the negation of the key value stored in the node. (Ranks will be non-negative, and we will only give nodes non-positive integer key values; these can be assumed to be perturbed arbitrarily to give a total ordering of key values). The rank of a node can increase as the result of a pairing, and the value of a node can decrease as the result of a \opDc. It is thus our goal to perform a \opDc\ on a node which has had its rank increase to restore it to the negation of its rank. During the time between when a rank increase occurs in a node and the time the \opDc\ is performed, we refer to the node as \emph{marked}. 

Call the \emph{unmarked subtree} of a node to be the subtree of a node if all marked nodes were detached from their parents; the \emph{unmarked structure} of the heap is the structure of the unmarked subtrees of the roots. The rank of a node at a given time will be defined to be a function of the structure of its unmarked subtree.
\begin{fullonly}
 We emphasize that the notion of marking a node is for the purposes of the analysis only, such marks need not be stored.
\end{fullonly}

The following assumes a particular heap structure and marking, as the rank of a node is always defined with respect to the structure of the heap after executing a sequence of heap operations. 
\begin{fullonly}

\end{fullonly}
Let $\akey$ be the node we wish to compute the rank of. Let $\umc$ denote the number of unmarked children of $\akey$, and let $y_1,y_2,\ldots y_{\umc}$ denote these children numbered right-to-left (i.e., in the order which they became children of $\akey$).
\begin{fullonly}

\end{fullonly}
Let $\tau_i(x)$ be a subtree of $x$ consisting of $x$ connected to only the subtrees induced by $y_1, y_2, \ldots y_i$. We will define the function $\rank_i(x)$ as a function of $\tau_i(x)$. The rank of a  node, $\rank(x)$ is $\rank_k(x)$. 
\begin{fullonly}

\end{fullonly}
Each node $y_i$ may be labeled as \emph{efficiently linked} to its parent.
If $ \rank_{i-1}(x)-\rcd \leq \rank(y_i) \leq \rank_{i-1}(x)$, then $y_i$ is said to be efficiently linked to $x$. The case of $ \rank(y_i) > \rank_{i-1}(x)$ will never occur, as pairings will only happen among unmarked nodes, where the rank perfectly matches the negation of the key value.
\begin{fullonly}

\end{fullonly}
We will have the property that $\rank_i(x)$ is either $\rank_{i-1}(x)$ or $\rank_{i-1}(x)+1$; In the latter case, $y_i$ is called \emph{incremental}. 
\begin{fullonly}

\end{fullonly}
Given a node $y_i$, let $j$ be defined to be the index of the first incremental node in the sequence $\langle y_{i-1}, y_{i-2}, \ldots\rangle$; $j$ is defined to be 0 if there is no such incremental node.
The set $\noninc(y_i)$ is defined to be $\{ y_k | j<k \leq i\}$; that is, $y_i$ and the maximal set of its non-incremental siblings to the right.
\begin{fullonly}

\end{fullonly}
Given these preliminaries, we can now give the full definition of the rank of a node that has defined rank:
\begin{fullonly}
$$
\rank_i(x)=
\left\{
	\begin{array}{ll}
		0  & \mbox{if } i = 0 \\ \\
		r_{i-1}(x)+1 & \mbox{\parbox{4in}{(Efficient case) $y_i$ is efficiently linked and is the $\rce $th efficient element of $\noninc(y_i)$
		\\%
		---or--- \\
		(Default case) $|\noninc(y_i)|=2^{\rcd} {}$}}\\ \\
		r_{i-1}(x) & \mbox{otherwise}
	\end{array}
\right.
$$
\end{fullonly}
\begin{shortonly}
$$
\rank_i(x)=
\left\{
	\begin{array}{ll}
		0  & \mbox{if } i = 0 \\ 
		r_{i-1}(x)+1 & \mbox{\parbox{4in}{(Efficient case) $y_i$ is efficiently linked and is the $\rce $th efficient element of $\noninc(y_i)$ or
		(Default case) $|\noninc(y_i)|=2^{\rcd} {}$}}\\ 
		r_{i-1}(x) & \mbox{otherwise}
	\end{array}
\right.
$$
\end{shortonly}
While the rank and mark are interrelated, there is no circularity in their definitions---whether a node is marked depends on its rank and key value and the rank of a node is a function of the ranks and marks of its children. 


\subsection{Useful facts about ranks}\jlabel{s:ufar}

\begin{obs}[Structural property of rank]\jlabel{lem:structrank}
Given two nodes $x$ and $y$ with different ranks, the unmarked structure of their induced subtrees must be different.
\end{obs}

This follows directly from the fact that the rank of a node is a function of its induced unmarked subtree.
\begin{fullonly}
\begin{lemma}The size of a unmarked subtree induced by a node of rank $k$ is at most $\rce ^k$.
\end{lemma}

\begin{proof}
Let $s_k$ be the size of the maximum unmarked heap of rank $k$. Such a heap can be created from a maximal unmarked heap of rank $k-1$, which has been paired to $\rce-1$ maximal unmarked heaps of rank $k-1$ and $2^{\rcd}{}-(\rce-1)$ maximal unmarked heaps of rank $k-\rcd-1$. Thus:

$$s_k \leq (\rce-1) s_{k-1}+2^{\rcd} s_{k-\rcd-1}.$$

By induction, 

$$s_k \leq (\rce-1) \rce ^{k-1}+2^{\rcd} \rce ^{k-\rcd-1}
=
\rce ^{k} + \rce ^{k-1}(2^{\rcd} \rce ^{-\rcd}-1).$$ 

For sufficiently large $n$, $2^{\rcd} \rce ^{-\rcd}\leq 1$ (recall that $\rce=2 \fdc+1$ and $\fdc=\omega(1)$), thus

$$ \rce ^{k} + \rce ^{k-1}(2^{\rcd} \rce ^{-\rcd}-1)\leq  \rce ^{k} $$

which completes the lemma. 

\end{proof}

\begin{cor} \jlabel{c:rootrank}
If there are $m$ nodes in a node $x$'s unmarked induced subtree, the rank of $x$ is at least $\log_{\rce} {m} $.
\end{cor}

\begin{lemma}[{Number of efficiently linked children}]
Suppose node $v$ has rank $\geq k$ and at most $ \frac{k}{2}2^{\rcd} {}$ unmarked children. Then, $v$ has at least $k/2$ efficiently linked unmarked children, each having rank $<k$.
\end{lemma}

\begin{proof} If there were less than $k/2$ efficiently linked children, than at least $k/2$ rank increments would be caused by the default case of the definition of rank and thus there would be at least $\frac{k}{2}{2}^{\rcd} {}$ unmarked children, a contradiction.
\end{proof}

\begin{cor}\jlabel{cor:effcountidea}
Consider a root with umarked subtree size $\rootsize$. 
The root has rank $\geq \log_{\rce } {\rootsize} $ by Corollary~\jref{c:rootrank}. Suppose it has $\leq \frac{\log_{\rce }{\rootsize}}{2}2^{\rcd} {}  $ unmarked children. Then it has $\geq \frac{1}{2} \log_{\rce } \rootsize$ efficiently linked unmarked children.
\end{cor}

Observe that:

\begin{align*}
 \frac{\log_{\rce }{\rootsize}}{2}2^{\rcd} {}
& = \frac{2^{\rcd-1}}{\log \rce}\cdot{ \log m }
\\ \intertext{and}
\frac{1}{2} \log_{\rce } \rootsize &=
\frac{1}{2 \log \rce}\cdot {\log m}
\end{align*}

We now use these observations to restate Corollary~\jref{cor:effcountidea}. 
\end{fullonly}
Set $\rcf=\rcfvalue{\rcd}{\rce}$ and $\rcg=\rcgvalue{\rce}$. 
\shortfull{Using some technical lemmas in \S\ref{full:s:ufar} the appendix gives:}{Then:}

\begin{cor}\jlabel{cor:effcount}
Suppose a root with unmarked subtree of size $\rootsize$ has $\leq {\rcf  \log \rootsize}$ children. Then it has $\geq  {\rcg \log \rootsize}$ efficiently linked children. 
\end{cor}

We note that $\rcf =\rcfvalue{\rcd}{\rce}=\frac{2^{2\fdc}}{\log (2 \fdc+1)}=o(\log n )$ since $\fdc=o(\log \log n)$.

\subsection{Monotonic operation sequences} \jlabel{s:monotonic}

Call the \emph{designated minimum root} the next node to be removed in an \opEm.
\begin{fullonly}
\begin{definition}[Monotonic operation sequence] \jlabel{def:monotonic}
\end{fullonly}
Define a \emph{monotonic operation sequence} to be one where \opDc\ operations are only performed on roots, children of the designated minimum root, or marked nodes.
\begin{fullonly}
\end{definition}
\end{fullonly}
All of the sequences of operations we define will be monotonic. We will need the following two observations\shortfull{}{ about monotonic sequences}:

\begin{obs}[{Monotonic sequences and structure}] In a monotonic operation sequence, for any node $x$ with descendent $y$ where all nodes on the path from $x$'s child down to and including  $y$ are unmarked, $y$ will remain in the same location in $x$'s subtree until $x$ becomes the designated minimum root.
\end {obs}

\begin{obs}[{Monotonic sequences and rank}]\jlabel{obs:monotonerank}
In a monotonic operation sequence, the rank of a node never decreases, from the time it is inserted until the time it becomes the designated minimum root.
\end{obs}

\subsection{Augmented suboperation} \jlabel{s:augmented}

We augment \shortfull{}{suboperation~\jref{so:pair},} the \subop{Pair}\cdot\ operation, to return whether or not the rank was incremented as a result of the pairing. This augmentation does not give any more power to the pure heap model, since the exact ranks of all nodes is a function of the suboperation sequence and is thus known to the algorithm already.
We use this augmentation to create a finer notion of what constitutes a distinct sequence of suboperations. In particular, we will use the following fact:

\begin{lemma} \jlabel{lem:subdistinct}
Suppose $s_i$ and $s_j$ are two structurally distinct states of the data structure. Suppose a single valid sequence of suboperations implementing an \opEm\ is performed on both, and the outcomes of all augmented suboperations that have return values are identical in both structures. Then, the position of all nodes who have had their ranks changed is identical in both.
\end{lemma}

\begin{fullonly}
\begin{proof}
The only time a node can have it's rank change is when something is paired to it, and this is now explicitly part of the return value of operation~\jref{so:pair}.
\end{proof}
\end{fullonly}

\subsection{Evolutions of indistinguishable sequences} \jlabel{s:evolution}

\subsubsection{Definitions} \jlabel{s:evolutiondefinition}

Let $\opB=\langle \opb_1, \opb_2, \ldots \rangle$   be a sequence of priority queue operations.
\begin{fullonly}

\end{fullonly}
Let $\opA(\opb_i)=\langle \opa^i_1, \opa^i_2, \ldots \rangle$ be the sequence of augmented suboperations and their return values used by algorithm \alg\ to execute operation $\opb_i$ if $\opb_i$ is an \opEm ; if it is not $\opA(\opb_i)$ is defined to be the empty sequence. $\opA(\opB)$ is the concatenation of $\opA(\opb_1), \opA(\opb_2), \ldots$.
\begin{fullonly}

\end{fullonly}
We call two sequences of priority queue operations $B$ and $B'$ \emph{algorithmically indistinguishable} if $A(B) =A(B')$, else they are \emph{algorithmically distinct}.
\begin{fullonly}

\end{fullonly}
Let $\state_{\opB}(i)$ be the structure of the heap after running sequence $\langle \opb_1\ldots \opb_i\rangle$; the \emph{terminal structure} of $B$ is $\state_{\opB}(|\opB|)$ which we denote as $\state_{\opB}$. Recall that by structure, we mean the raw shape of the heap without regard to the data in each node, but including which nodes are marked. Two sequences $\opB$ and $\opB'$ are \emph{terminal-structure indistinguishable} if $\state_B = \state_{B'}$, else they are \emph{terminal-structure distinct}.
Given a set of mutually algorithmically indistinguishable and terminal-structurally distinct (AI-TSD) sequences of heap operations $\Ids$, the \emph{distinctness} of the set, $\Distinct(\Ids)$ is defined to be $\log |\Ids|$. 
\begin{fullonly}

\end{fullonly}
{Note that having two sequences which are algorithmically indistinguishable does not imply anything about them being terminal-structure indistinguishable. For example, it may be possible to add a \opDc\ to a sequence, changing the terminal structure, while the sequence of suboperations performed to execute the sequence remains unchanged. (Recall that suboperations only occur during \opEm\ operations).}
\begin{fullonly}

\end{fullonly}
A critical observation is that the number of terminal-structurally distinct sequences is function of $n$; this is the basis for the contradiction at the end of the proof:

\begin{lemma} \jlabel{lem:maxdistinct}
The maximum distinctness of any set $\Ids$ of terminal-structurally distinct sequences, all of which have terminal structures of size $n$, is $\Distinct(\Ids)=O(n)$.
\proofapp{lem:maxdistinct}
\end{lemma}

\begin{fullonly}
\begin{proof}
The number of different shapes of a rooted ordered forest with $n$ nodes is $C_n$, the $n$\superscript{th} Catalan number. The number of different ways to mark some nodes in a forest with $n$ nodes is $2^n$. Since $C_n \leq 4^n$, the maximum number of distinct structures is at most $\log (2^n4^n) = O(n)$. 
\end{proof}
\end{fullonly}

\subsubsection{Evolving} \jlabel{s:evolving}

We will now describe several functions on AI-TSD sets of heap operations; we call such functions \emph{evolutions}. The general idea is to append individual heap operations or small sequences of heap operations to all sequences in the input set $\Ids$ and remove some of the resulting sequences so as to maintain the property that the sequences in the resultant set of sequences $\Ids'$ are algorithmically indistinguishable yet terminal-structure distinct. The evolutions will also have the property that if the time to execute all sequences in $\Ids$ is identical, then the runtime to execute all sequences in $\Ids'$ will also be identical. The difference in the runtime to execute sequences in $\Ids'$ versus those in $\Ids$ will be called the \emph{runtime} of an evolution.

\subsubsection{Insert evolution} \jlabel{s:einsert}

The \emph{insert evolution} has the following form: $\Ids' = \eIns(\Ids)$.
\begin{fullonly}

\end{fullonly}
In an insert evolution, a single \opIns\ operation of a key with value 0 is appended to the end of all $\Xi$ to obtain $\Xi'$. Given $\Ids$ is AI-TDS, the set $\Ids'$ is AI-TDS and trivially $\Distinct(\Ids)=\Distinct(\Ids')$. The runtime of the evolution is 1 since the added \opIns\ has runtime 1. The rank of the newly inserted node is 0, and is thus unmarked.

\subsubsection{Decrease-key evolution} \jlabel{s:edc}

The \emph{decrease-key evolution} has the following form: $\Ids' = \eDc(\Ids,p)$, where $p$ is a location which is either a root or a marked node in all terminal structures of sequences in $\Ids$. 
\begin{fullonly}

\end{fullonly}
In a \emph{\opDc\ evolution}, a \opDc$(p, \Delta x)$ operation is appended to the end of all sequences in  $\Xi$ to obtain $\Xi'$. 
The value of $\Delta x$ is chosen such that the new key value of what $p$ points to is set to is the negation of its current rank; this means $\Delta x$ is always nonnegative because of the monotone property of ranks noted in Observation~\jref{obs:monotonerank}. 
Observe that if $p$ points to a marked node, then it is unmarked after performing a \eDc. This requirement ensures that all structures that are distinct before this operation will remain distinct after the operation. Thus, the set $\Ids'$ is AI-TDS and trivially $\Distinct(\Ids)=\Distinct(\Ids')$. The runtime of the evolution is 1 since the added \opDc\ has runtime 1.

\subsubsection{Designated minimum root evolution} \jlabel{s:edmr}

The \emph{designated minimum root evolution} has the form $\Ids'=\eDmr(\Ids,r)$, where $\bigroot$ is the position of one root which exists in all terminal structures of $\Ids$.
\begin{fullonly}

\end{fullonly}
In a designated minimum root evolution, a \opDc\ operation on $\bigroot$ to a value of negative infinity is appended to all sequences in $\Ids$ to give $\Ids'$. It will always be the case that the (next) evolution performed on $\Ids'$ will be an \eEm\ evolution; the root $\bigroot$, which is known as the \emph{designated minimum root}, will be removed from all terminal structures of $\Ids'$ in this subsequent \eEm. There is no change in distinctness caused by this operation: $\Distinct(\Ids)=\Distinct(\Ids')$. The runtime of the evolution is 1 since the added \opDc\ has runtime 1.

\subsubsection{Extract-min evolution} \jlabel{s:em}

The \emph{extract-min evolution} has the form $(\Ids',\Vio,\acem)=\eEm(\Ids)$.
\begin{fullonly}

\end{fullonly}
The \eEm\ evolution is more complex than those evolutions previously described, and the derivation of $\Ids'$ from $\Ids$ is done in several steps. First, a \opEm\ operation is appended to the end of all sequences in $\Ids$ to obtain an intermediate set of sequences which we call $\tent{\Ids}$. Recall that the \opEm\ operation is implemented by a number of suboperations. There is no reason to assume that the suboperations executed by the algorithm in response to the \opEm\ in each of the elements of $\tent{\Ids}$ are the same; thus the set  $\tent{\Ids}$ may no longer be algorithmically indistinguishable. We fix this by removing selected sequences from the set $\tent{\Ids}$ so that the only ones that remain execute the appended \opEm\ by using identical sequences of supoperations. This is done by looking at the first suboperation executed in implementation of \opEm\ in each element of $\tent{\Ids}$, seeing which suboperation is the most common, and removing all those sequences  $\tent{\Ids}$ that do not use the most common first suboperation. If the suboperation is one which has a return value, the return value which is most common is selected and the remaining sequences are removed. This process is repeated for the second suboperation, etc., until the most common operation is \subop{End}{} and thus the end of all remaining suboperation sequences has been simultaneously reached. Since there are only a constant $\subops$ number of suboperations, and return values, if present, are boolean, at most a constant fraction (specifically $\frac{1}{2 \subops}$) of $\tent{\Ids}$ is removed while pruning each suboperation. At the end of processing each suboperation by pruning the number of sequences, the new set is returned as $\Ids'$. The set $\Ids'$ can be seen to be terminal-structure distinct, since pairing identically positioned roots in structurally different heaps, and having the same nodes win the pairings, can not make different structures the same. 

Observe that the nodes winning pairings in the execution of the $\opEm$ might have their ranks increase, and thus become marked. \shortfull{By}{Moreover, due to} Lemma~\jref{lem:subdistinct}, the position of all such nodes is identical in all terminal structures of $\Ids'$. The set of the locations of these newly marked  nodes is returned as $\Vio$, the \emph{violation set}.

Now that it has been ensured that all of the sets of operations execute the appended \opEm\ using the same suboperations, we 
define $\acem$ to be this common number of suboperations used to implement the \opEm; this value is returned by the evolution. As each suboperation reduces the distinctness by at most a constant,
$\Distinct(\Ids') \geq \Distinct(\Ids)-\acem \log (2 \subops)=  \Distinct(\Ids)- O(\acem)$. The runtime of the evolution is $\acem$ since that is the cost of the added \opEm.

\subsubsection{Big/small evolution} \jlabel{s:bs}
%
%
%
The \emph{big/small evolution} has the form $(\Ids',(p,bigsmall))=\eBs(\Ids)$.
\begin{fullonly}

\end{fullonly}
The goal of the big/small evolution is to ensure that the terminal structures of all sets are able to be executed in the same way in subsequent evolutions.
In a big/small evolution, each terminal structure of each of the operation sequences of $\Ids$ is classified according to the following, using the previously-defined function $\rcf$: 
\begin{itemize\shortfull{*}{}}

\item The exact number of roots if less than $\rcf  \log n$ or the fact that the number of roots is greater than $\rcf  \log n$ (we call this case \emph{many-roots}). 

\item If the exact number of roots is less than ${\rcf  \log n}$:
\begin{fullonly}
\begin{itemize}
\end{fullonly}
\shortfull{}{\item} The position of the root with the largest subtree (the leftmost such root if there is a tie). Call it $p$. Observe that the size of $p$'s subtree is at least 
$\cst$.
\shortfull{}{\item} The exact number of children of $p$ if less than ${\rcf  \log \cst}$ (we call this case \emph{small}) or the fact that the number of roots is greater than ${\rcf {} \log \cst}$ (we call this case \emph{root-with-many-children}).

\begin{fullonly}
\end{itemize}
\end{fullonly}

\end{itemize\shortfull{*}{}}

There are at most $\lceil {\rcf {} \log n} \rceil \cdot \lceil {\rcf {} \log \cst} \rceil$ possible classifications. We create set $\Ids'$ by removing from $\Ids$  sequences with all but the most common classification of their terminal structures. 
The return value is based on the resultant classification:
\shortfull{}{\begin{description}}
\shortfull{\emph{Many-roots:}}{\item[Many-roots:]} Return $(p,bigsmall)$ where $p=\text{\sc{NULL}}$ and $bigsmall=\text{\sc{Big}}$.
\shortfull{\emph{}Root-with-many-children:}{\item[Root-with-many-children:]} Return $(p,bigsmall)$ where $p$ is the location of the root with the largest subtree and $bigsmall=\text{\sc{Big}}$.
\shortfull{\emph{Small:}}{\item[Small:]} Return $(p,bigsmall)$ where $p$ is the location of the root with the largest subtree and and $bigsmall=\text{\sc{Small}}$.
\shortfull{}{\end{description}}

We bound the loss of distinctness, which is the logarithm of the number of classifications.
 Since $\rcf=o(\log n)$, then $\log \left( \lceil {\rcf {} \log n} \rceil \cdot \lceil {\rcf {} \log \cst} \rceil\right)= O(\log \log n)$, and thus $\Distinct(\Ids') = \Distinct(\Ids)- O(\log \log n)$. The runtime of the evolution is 0 since no operations are added to any sequence.

\subsubsection{Permutation evolution} \jlabel{ev:perm}

The \emph{permutation evolution} has the form $\Ids'=\ePerm (\Ids )$, where the leftmost root $\bigroot$ has in all terminal structures of the sequences of $\Ids$  a subtree size of at least $\frac{n  }{\rcf \log n}$ and at most ${\rcf  \log \cst}$ children; this will be achieved by being in the  small case of the big/small evolution and performing a decrease-key evolution on the relevant node.
It is also required that all terminal structures of sequences in $\Ids$ are entirely unmarked.
\begin{fullonly}

\end{fullonly}
In a permutation evolution, the goal is to increase the distinctness, and is the only evolution to increase the number of sequences\shortfull{.}{ in the process of converting $\Ids$ to $\Ids'$.} 

\shortfull{By}{Combining the preconditions of the permutation evolution with} corollary~\jref{cor:effcount}\shortfull{}{, yields the fact that} all nodes in the terminal structures of $\Ids$ at location $\bigroot$ have at least ${\rcg \log \cst}$ efficiently linked children; since there are at most $\rce $ efficiently linked children of each rank, \shortfull{}{that means} there are at least  $\frac{\rcg \log \cst}{\rce   }$ efficiently linked children of different ranks in each terminal structure. Find such a set and call it the \emph{permutable set\shortfull{ (PS)}{}}\shortfull{.}{(chose one arbitrarily if there is more than one possibility).} \shortfull{P}{Look at the position of these permutable sets in all terminal structures of the sequences of $\Ids$, and p}ick the position of the \shortfull{PS}{permutable set} that is most common. Form the intermediate set of sequences $\hat{\Ids}$ by removing from $\Ids$ all sequences that do not have this commonly located \shortfull{PS}{permutable set}.
An easy upper bound on the number of different locations \shortfull{PS}{permutable set}s could be in is
\shortfull{$}{$$} \overbrace{\frac{\rcg \log \cst}{\rce   }}^{\text{Size of permutable set}} \cdot  \overbrace{{\rcf  \log \cst}}^{\text{Number of children of $r$}}. \shortfull{$}{$$}
\begin{shortonly}
The reduction of distinctness is the logarithm of this, which, as $\rcf=o(\log n)$, $\rcg=o(1)$, $\rce=\omega(1)$, is
$\Distinct(\hat{\Ids})-\Distinct(\Ids)= -O(\log \log n)$.
\end{shortonly}
\begin{fullonly}
Thus this step decreases the distinctness of the set by at most the logarithm of the number of commonly located \shortfull{PS}{permutable set}s:  
\shortfull{$}{$$}\Distinct(\hat{\Ids})-\Distinct(\Ids)=   - 
O\left( \log \left( \frac{\rcg \log \cst}{\rce   } \cdot  {\rcf  \log \cst} \right) \right).\shortfull{$}{$$}
As $\rcf=o(\log n)$, $\rcg=o(1)$, $\rce=\omega(1)$, one can simplify this to
\shortfull{$}{$$}\Distinct(\hat{\Ids})-\Distinct(\Ids)=   -O(\log \log n).\shortfull{$}{$$}
\end{fullonly}

\shortfull{Using the definitions of $\rcf$ and $\fdc$, the \shortfull{PS}{permutable set} is of size}{The \shortfull{PS}{permutable set} is of size $\frac{\rcg \log \cst}{\rce   }$. Using the definitions of $\rcf$ and $\fdc$, this is} $\Theta\left(\frac{\log n}{\fdc \log \fdc}\right)$. Let $\rcm$ be a constant such that the \shortfull{PS}{permutable set} is of size at least $\frac{\rcm \log n}{\fdc \log \fdc}$ for sufficiently large $n$. Recall that in all the terminal structures
of the sequences of $\hat{\Ids}$, the ranks of the items in the permutable sets are different, and in fact are strictly increasing, when viewed right-to-left as children of $r$.

 We then create $\Ids'$ by replacing each sequence in ${\hat{\Ids}}$ with $\left(\frac{\rcm \log n}{\fdc \log \fdc}\right)!$ new sequences created by appending onto the end of each existing sequence a sequence of all possible permutations of \opDc\ operations on all elements of an arbitrary subset of size $\frac{m \log n}{\fdc \log \fdc}$ of the \shortfull{PS}{permutable set}.

The fact that all of the sequences in ${\hat{\Ids}}$ have the same \shortfull{PS}{permutable set}s ensures that all terminal structures in $\Ids'$ are terminal-structure distinct. 
 (Recall that Lemma~\jref{lem:structrank} says that different ranks imply different structures of induced subtrees). 
Thus, in this step distinctness increases by $\Distinct(\Ids')-\Distinct({\hat{\Ids}})=\log \left( \frac{\rcm \log n}{\fdc \log \fdc}  \right)! = \Theta\left( \frac{\log n \log \log n}{\fdc \log \fdc} \right)$\shortfull{, which dominates the total change of disctintness.}{.} 
\begin{fullonly}

Thus, combining all the steps of the permutation evolution bounds the total increase of distinctness by
\shortfull{$\Theta\left( \frac{\log n \log \log n}{\fdc \log \fdc} \right) $.}%
{
\begin{align*}
\Distinct(\Ids')-\Distinct(\Ids)&=
\Distinct(\Ids')-\Distinct(\hat{\Ids})+\Distinct(\hat{\Ids})-\Distinct(\Ids)\\
&=\Theta\left( \frac{\log n \log \log n}{\fdc \log \fdc} \right) -O(\log \log n)\\
&=\Theta\left( \frac{\log n \log \log n}{\fdc \log \fdc} \right) 
\end{align*}
}

\end{fullonly}
\shortfull{The evolution costs at most $\frac{\rcm \log n}{\fdc \log \fdc}$, since that is the number of unit-cost \opDc\ operations appended to the sequences.}%
{The cost of the evolution is at most $\frac{\rcm \log n}{\fdc \log \fdc}$, since that is the number of \opDc\ operations appended to the sequences, and these all have unit cost.}

\subsection{Rounds} \jlabel{s:rounds}

\begin{fullonly}
\begin{algorithm}
\caption{Algorithmic
 presentation of how evolutions are used to build the sequence of AI-TSD sequences 
 $\langle \Xi_0, \Xi_1, \Xi_2, \ldots \rangle$, which are split into rounds where the index of the start of round $i$ is $\circ_i$.
} 
\jlabel{a:evolve}
\begin{algorithmic}
\State $\Ids_0=\{\langle \overbrace{\opIns(0), \opIns(0), \ldots, \opIns(0)}^{n\text{ \opIns\ operations}} \rangle \}$
\State $i=0$
\State $round=0$
\State $\round_0=0$;
\Loop
\State $(\Ids_{i},p,bigsmall)= \eBs(\Ids_{i \pp})$; 
\If {$bigsmall=\text{\sc Small}$} 
\State $\Ids_{i} = \eDc(\Ids_{i \pp},p)$; \Comment{Small round}
\State $\Ids_{i}=\ePerm(\Ids_{i \pp})$; 
\Else 
\State $\Ids_{i}=\eDmr(\Ids_{i \pp},p)$; \Comment{Big round}
\EndIf
\State  $(\Ids_{i},\Vio,\acem)=\eEm(\Ids_{i \pp})$; \Comment{Common to small and big rounds}
\For{each $v$ in $\Vio$}
\State $\Ids_{i} = \eDc(\Ids_{i \pp},v)$;
\EndFor
\State $\Ids_{i} = \eIns(\Ids_{i \pp})$; 
\State $\round_{\pp round}=i$; 
\EndLoop
\end{algorithmic}

\end{algorithm}
\end{fullonly}

\shortfull{A sequence of evolutions $\Evo=\langle \evo_0, \evo_1,\ldots \rangle$ defines}
{We proceed to perform a sequence of evolutions $\Evo=\langle \evo_0, \evo_1, \ldots \rangle$ to define} a sequence of AI-TSD sets $\langle \Ids_0, \Ids_1, \ldots \rangle$. The initial set $\Ids_0$ consists of a single sequence of operations: the operation \opIns$(0)$, executed $n$ times. Each subsequent AI-TSD  set $\Ids_{i}$ is derived from  $\Ids_{i-1}$  by performing the single evolution $\evo_{i-1}$; thus in general $\Ids_{i}$ is composed of some of the sequences of $\Ids_{i-1}$ with some operations appended. 

These evolutions are split into \emph{rounds}; $\round_i$ is the index of the first AI-TSD set of the $i$th round. Thus round $i$ begins with AI-TSD set $\Ids_{\round_i}$ and ends with $\Ids_{\round_{i+1}-1}$ through the use of evolutions
$\langle \evo_{\round_{i}} \ldots \evo_{\round_{i+1}-1}  \rangle$
\begin{fullonly}

\end{fullonly}
 These rounds are constructed to maintain several invariants:
\shortfull{}{\begin{itemize}\item}
All terminal structures of all sequences in the AI-TSD set at the beginning and end of each round have size $n$. This holds as in each round, exactly one \opIns\ evolution and exactly one \opEm\ evolution is performed.
\shortfull{}{\item} All nodes in all terminal structures in the AI-TSD sets at the beginning and end of each round are unmarked.
\shortfull{}{\end{itemize}}

There are two types of rounds, \emph{big rounds} and \emph{small rounds}. 
At the beginning of both types of round a big/small evolution is performed\shortfull{ which }{. The return value of the big/small evolution} determines \shortfull{the round type.}{whether it will be a big or a small round.}
\begin{fullonly}

\end{fullonly}
The reader may refer to Algorithm~\ref{full:a:evolve} 
\begin{shortonly}
in the Appendix
\end{shortonly}
for a concise presentation of how evolutions are used to construct $\langle \Ids_0, \Ids_1, \Ids_2, \Ids_3, \ldots \rangle$. \shortfull{}{We now describe this process in detail.}

\subsubsection{The Big Round} \jlabel{s:br}

As the round begins, we know that the terminal structures of the AI-TSD set are entirely unmarked, and there are either at least ${\rcf  \log n}$ roots, or one root with at least ${\rcf {} \log \cst}$ children.  The round proceeds as follows:
\begin{fullonly}

\end{fullonly}
\begin{shortonly}
\begin{inparaenum}[(1)]
\end{shortonly}
\begin{fullonly}
\begin{enumerate}
\end{fullonly}
\item Perform a designated minimum root evolution on the root with largest subtree; this is the node $\bigroot$ from the big-small evolution whose location is encoded in the return value; as a result of the big-small evolution it is guaranteed to be in the same location in all of the terminal structures of the sequences of $\Ids$. 
\item Perform a \opEm\ evolution. 
\item For each item in the violation sequence returned by the \opEm\ evolution, perform a \opDc\ evolution. This makes the terminal structures of all heaps in $\Ids$ mark-free. 
\item Perform an \opIns\ evolution. 
\begin{shortonly}
\end{inparaenum}
\end{shortonly}
\begin{fullonly}
\end{enumerate}
\end{fullonly}

Assuming we are in round $i$, Let $\acem_i$ be the cost of the \opEm\ evolution, and let $\vs_i$ be the size of the violation sequence.
The cost of the round (the sum of the costs of the evolutions) is $\acem_i + \vs_i+2$, which is at least ${\rcf {} \log \cst}$, and based on the evolutions performed the distinctness can be bounded as follows: 
\shortfull{$}{$$}\Distinct_i-\Distinct_{i+1}=\overbrace{O(\acem_i)}^\opEm +\overbrace{O(\log \log n)}^\eBs . \shortfull{$}{$$}


\subsubsection{The Small Round} \jlabel{s:sr}

\shortfull{T}{In a small round, t}here is one root, call it $\akey$, at the same location in all terminal structures, with size at least $\frac{n  }{\rcf   \log n}$ and some identical number of children in all terminal structures which is at most ${\rcf {} \log \cst}$. The location of $\akey$ was returned by the big-small evolution. The round proceeds as follows:
\begin{shortonly}
\begin{inparaenum}[(1)]
\end{shortonly}
\begin{fullonly}
\begin{enumerate}
\end{fullonly}
\item Perform an \opDc\ evolution on $\akey$ to make it negative infinity. 
\item  Perform a \ePerm\ evolution. 
\item Perform an \opEm\ evolution. 
\item For each item in the violation sequence returned by the \opEm\ evolution, perform a \opDc\ evolution. 
\item Perform an \opIns\ evolution. 
\begin{shortonly}
\end{inparaenum}
\end{shortonly}
\begin{fullonly}
\end{enumerate}
\end{fullonly}

Let $\acem_i$ be the actual cost of the \opEm,  let $\vs_i$ be the size of the violation sequence.
\begin{fullonly}

\end{fullonly}
The cost of the round is $\acem_i + \vs_i+2+\frac{\rcm \log n}{\fdc \log \fdc}$, and based on the evolutions performed the distinctness can be bounded as follows: 
%
\shortfull{$}{$$}\Distinct_i-\Distinct_{i+1}=\overbrace{O(\acem_i)}^\opEm +\overbrace{O(\log \log n)}^\eBs-\overbrace{\Omega\left( \frac{\log n \log \log n}{ \fdc \log \fdc} \right)}^\ePerm\shortfull{$}{$$}


\subsection{Upper bound on time} \jlabel{s:ubot}

\begin{fullonly}
The following two crucial lemmas are needed in the next section.
\end{fullonly}
\begin{shortonly}
The following two crucial lemmas are needed in the next section and have proofs in \S\ref{full:s:ubot} in the appendix. 
\end{shortonly}
\begin{lemma}\jlabel{lem:totaltime}
The total time to execute any sequence in $\Ids_{\round_k}$ is $O(k \log n)$.
\end{lemma}

\begin{fullonly}

\begin{proof}
Let $\opB$ be a sequence in $\Ids_{\round_k}$, and let $\runb$ be the time to execute $B$.

Let $\dcc_i$ be the cost of the permutation evolution in round $i$; this is at most $ \frac{\rcm \log n}{\fdc \log \fdc}$
 if round $i$ is a small round and 0 if round $i$ is a big round (recall that permutation evolutions are only performed in the small round).
Thus the cost for any round $i$, whether big (\S\jref{s:br}) or small (\S\jref{s:sr}), can be expressed as $2+\dcc_i +    \acem_i +  \vs_i$.
The time to execute any sequence $\opB \in \Ids_{\round_k}$, which we denote as \runb, is the sum of the costs of the rounds:

\begin{equation} \jlabel{eq1}
 \runb =  \sum_{i=1}^k (2+\dcc_i +    \acem_i +  \vs_i) 
 \end{equation}

An item in the violation sequence has had its rank increase. Its rank can only increase after winning $\rce$ pairings. Pairings are operations. Thus,

\begin{align}
 \sum_{i=1}^k \vs_i 
 &\leq \frac{\runb}{\rce }  \jlabel{eq2}
\\
\intertext{Combining~\jeqref{eq1} and~\jeqref{eq2} and rearranging gives:}
\sum_{i=1}^k \vs_i &\leq \frac{1}{\rce}{\sum_{i=1}^k (2+\dcc_i +    \acem_i +  \vs_i)}
\nonumber
\\
\left(1-\frac{1}{\rce} \right)
\sum_{i=1}^k \vs_i &\leq \frac{1}{\rce}{\sum_{i=1}^k (2+\dcc_i +    \acem_i )}
\nonumber
\\
\sum_{i=1}^k \vs_i &\leq  \frac{1}{\rce -1}\sum_{i=1}^k (2+\dcc_i +   \acem_i )
\jlabel{eq3}
\end{align}

We know by assumption that the runtime is bounded by the sum of the amortized costs stated at the beginning of \S\jref{sec:sor}. Each round has one \opIns\ and one \opEm\ (at an amortized cost of $\cem \log n$ each) and $1+dc_i + \vs_i$ \opDc\ operations (at an amortized cost of $\LL{\fdc}$ each). This gives:

\begin{equation} \jlabel{eq4}
\runb\leq \sum_{i=1}^k (  (\dcc_i + \vs_i +1 ) {\fdc}  + 2 \cem \log n )
\end{equation}


\begin{align*}
\intertext{Combining \jeqref{eq3} and \jeqref{eq4} gives:} \jlabel{eq5}
\runb & \leq \sum_{i=1}^k \left[ \left(  \dcc_i  +
\frac{2+\dcc_i +   \acem_i }{\rce -1} +1\right) \fdc
   + 2 \cem \log n \right]
\\
\intertext{Since $\dcc_i \leq \frac{\rcm\log n}{\fdc \log \fdc}$ and 
 $\frac{\rcm\log n}{\fdc \log \fdc} \leq \frac{\log n}{\fdc}$ for sufficiently large $n$ (recall that $\rcm$ is a constant defined in \S\jref{ev:perm}).}
\runb &\leq \sum_{i=1}^k \left[ \left( 
\frac{2 }{\rce -1} +1 \right) \fdc +
 \frac{\acem_i \fdc}{\rce-1} 
   + \left( 1+2 \cem+\frac{1}{\rce-1} \right) \log n \right]
\\
\runb&\leq 
k \log n
\left[
\left(  \frac{2 }{\rce -1}+1 \right) \frac{\fdc}{\log n} 
   + \left( 1+2 \cem+\frac{1}{\rce-1} \right) 
\right]
+
\frac{\fdc}{\rce-1}
\sum_{i=1}^k 
\acem_i 
\\
\intertext{Since $\sum_{i=1}^k 
   \acem_i \leq \runb$}
\runb&\leq 
k \log n
\left[
\left(  \frac{2 }{\rce -1} +1\right) \frac{\fdc}{\log n} 
   + \left( 1+2 \cem+\frac{1}{\rce-1} \right) 
\right]
+
\frac{\fdc}{\rce-1}
\runb
\\
\runb
\left(
1-\frac{\fdc}{\rce-1}
\right)
&\leq 
k \log n
\left[
\left(  \frac{2 }{\rce -1} +1\right) \frac{\fdc}{\log n} 
   + \left(1+ 2 \cem+\frac{1}{\rce-1} \right) 
\right]
\\
\intertext{Substituting in the definition of $\rce$: $\rce=2\fdc+1$}
\runb
\left(
1-\frac{1}{2}
\right)
&\leq 
k \log n
\left[
\left(  \frac{1 }{\fdc}+1 \right) \frac{\fdc}{\log n} 
   + \left(1+ 2 \cem+\frac{1}{2\fdc} \right) 
\right]
\\
\runb
&\leq 
2k \log n
\left[
\frac{1 }{\log n} 
+
\frac{\fdc}{\log n} +
1+
2 \cem+
    \frac{1}{2\fdc} 
\right]
\intertext{For large enough $n$, $\frac{\fdc}{\log n} +
\frac{1 }{\log n} 
   + 
    \frac{1}{\fdc}<1$, so}
\runb
&\leq 
2(2\cem+2)k \log n
\end{align*}

Since $\cem$ is a constant, this concludes the proof.
\end{proof}

\end{fullonly}

\begin{lemma}\jlabel{lem:fracsmall}
After $k$ rounds, at least $\frac{k}{2}$ rounds must be small rounds.
\end{lemma}

\begin{fullonly}
\begin{proof}
Proof is by contradiction. Suppose more than $\frac{k}{2}$ rounds are big rounds.
The actual cost of a big round is at least ${\rcf {} \log \cst}$, so the total actual cost is greater than:

\begin{align*}
\runb & \geq \frac{k \rcf}{2} \log \cst \\
& = k \frac{2^{\rcd-1}}{2 \log \rce} \log \frac{n \log \rce}{2^{\rcd-1} \log n}
& \text{since } \rcf=\frac{2^{\rcd-1}}{\log \rce}\\
& =k \frac{2^{2\fdc}}{2 \log (2 \fdc+1)} \log \frac{n \log (2 \fdc +1)}{2^{2\fdc} \log n}
& \text{since } \rcd=2\fdc+1\\
& =  \Omega (k {2^{2\fdc-\log \log \fdc}} \log n)
\end{align*}

Since $\fdc=\omega(1)$, this contradicts Lemma~\jref{lem:totaltime}, that the runtime is $O(k \log n)$.
\end{proof}
\end{fullonly}

\subsection{Putting it together} \jlabel{s:pit}

The gain of distinctness of a round has been bounded as follows:
\begin{equation}
 \Distinct_{\round_{i+1}}-\Distinct_{\round_i} =
 \begin{cases}
-O(\acem_i)-O(\log \log n)  & \text{If the $i$th round is a big round (\S\jref{s:br})}\\
-O(\acem_i)-O(\log \log n) + \Theta \left( \frac{\log n \log \log n}{\fdc \log \fdc} \right)
&\text{If the $i$th round is a small round (\S\jref{s:sr})}
\end{cases}\jlabel{eq5}
\end{equation}
%
Now we know that $\sum_{i=1}^k (\acem_i)$ is less than the actual cost to execute a sequence in $\Ids_{\round_k}$, which is $O(k {} \log n )$ by Lemma~\jref{lem:totaltime}. Substituting $\sum_{i=1}^k \acem_i = O(k {} \log n )$ into \jeqref{eq5}  and using the fact from Lemma~\jref{lem:fracsmall} that at least half of the rounds are small rounds gives:
\shortfull{$}{\begin{equation}} \jlabel{eq6} \Distinct_{\round_k}-\Distinct_{\round_0} =  \Theta \left( k\log n \frac{\log \log n}{\fdc \log \fdc} \right)-O(k \log \log n)- O(k  \log n )  
\shortfull{$}{\end{equation}}
Since $\fdc=o\left( \frac{\log \log n}{\log \log \log n} \right)$, 
$\frac{\log \log n}{\fdc \log \fdc} =\omega(1)$, and thus
the negative terms in\shortfull{ the previous equation}{~\jeqref{eq6}} can be absorbed, giving:
\shortfull{$}{$$} \Distinct_{\round_k}-\Distinct_{\round_0} = \Theta \left( \frac{k\log n \log \log n}{\fdc \log \fdc} \right).
\shortfull{$}{$$}
But after sufficiently many rounds (i.e.~sufficiently large $k$) this contradicts Lemma~\jref{lem:maxdistinct} that for all $i$, $\Distinct_i = O(n)$. Thus for sufficiently large $k$ and $n$ a contradiction has been obtained, proving Theorem~\jref{th:main}.

\begin{fullonly}
\section{The sort heap}
\jlabel{s:sortheap}

\subsection{Purpose}

All heaps data structures that support $O(\log \log n)$ \opDc\ require either augmented data (e.g.~Fibonacci heaps and variants), or 
an implementation of \opDc\ that places them outside of both our and Fredman's \cite{DBLP:journals/jacm/Fredman99} models for lower bounds. In this section we describe a new and simple structure, which we call the \emph{sort heap}, which has no augmented data and for which our lower bound (but not Fredman's) applies. 
Our structure features $O(\log \log n)$ \opIns\ and \opDc, and $O(\log n \log \log n)$ \opEm.
When the ratio of \opDc\ operations to \opEm\ operations is $\Omega(\log n)$, \opDc\ is the dominant operation in the amortized runtime.

Thus the sort heap is the first self adjusting heap with $o(2^{2 \log \log n})$ \opDc\ in a model with lower bounds for \opDc; 
in our model the $O(\log \log n)$ amortized runtime for \opDc\ is within a $O(\log \log \log n)$ factor of the
$\Omega \left( \frac{\log \log n}{\log \log \log n} \right)$ amortized lower bound.
The sort heap only differs from the pairing heap in the choice of pairings performed in the \opEm\ operation.

The potential function used to analyze the sort heap is very different from that used to analyze Fibonacci heaps, pairing heaps, and their variants. The potential of a node is in the range $0 \ldots \Theta(\log \log n)$, which is a much smaller range than that used in all other potential functions. 
(A smaller range of a potential function indicates that the analysis can be applied to smaller sequences. The size of
a pontential function's range plays a role in some deamortization transformations; see~e.g.~\cite{DBLP:conf/icalp/BoseCFL12}).
Secondly, and perhaps more interestingly, the potential function is not logarithmic in nature. Previous functions were dominated by logarithms of some function of the subtree size of the node and perhaps the parent of the node. Our potential function has a dependence on subtree size, but it is linear. It is inspired by  the potential function used to analyze skew heaps \cite{DBLP:journals/jacm/SleatorT85}, a heap that predates all heaps with $o(\log n)$ $\opDc$. 

\subsection{Implementation of operations}

The data structure is a min-heap-ordered general tree. 

\begin{description}

\item[$p=$\opIns$(x)$:] Create a new node with key $x$ and adds it as a new leftmost heap in the forest of heaps. Returns a pointer $p$ to the new node.

\item[\opDc$(p,\Delta x)$:] Detaches the node pointed to by $p$ from its parent (if it is not a root), decreases the key value by the nonnegative value~$\Delta x$, and places it as the root of the leftmost heap in the forest of heaps.

\item[$x=$\opEm$()$:] 
 Let $k$ be the total number of roots in the tree. These $k$ roots now must be combined using the pairing operation. These roots are grouped into $\left \lceil \frac{k}{\consd d \log n} \right \rceil$ groups of $ \consd \log n$  nodes\footnote{The constant \consd\ is chosen for convenience; any constant strictly larger than $\frac{2}{\log \frac32}$ will work as will be seen in the proof of Lemma~\jref{l:em}.}; the last group might not be full. Within each group, the key values in the roots are sorted and then paired in order from largest to smallest. This leaves one heap per group; these heaps are combined via pairings of their roots arbitrarily until a single heap remains. Finally, now the the structure is a single heap, the root is removed and returned, leaving behind a forest of heaps with its children as the new roots.
 
\end{description}

\paragraph{Comparison to the Fibonacci heap.} 
The primary difference between the sort heap and the Fibonacci heap (and the pointer Fibonacci heap) is that in the implementation of \opEm, Fibonacci heaps sort according to subtree size and not according to key value. In the Fibonacci heap, a rough approximation of logarithm of the subtree size is stored explicitly as an augmented field (called the \emph{rank}, whose definition is not the same as the use of the word rank in this work). Fibonacci heaps also have the notion of \emph{marked nodes} and \emph{cascading cuts}; these notions add complexity to the implementation of the \opDc\ operation in some cases.

\paragraph{Comparison to the pairing heap.} 
The only difference between the sort heap and the pairing heap is the choice of pairings performed in \opEm. Otherwise, they are basically identical.

\paragraph{Comparison to Elmasry's heap.} 
The only difference between the sort heap and Elmasry's heap is the implementation of \opDc. Elmasry's heap has an implementation of \opDc\ where the leftmost child of the node having its key decreased is removed and is attached in the place of its former parent. The sort heap does not have this leave-one-child-behind behavior; all children of the node being decrease-keyed remain attached.

So, the sort heap takes the idea of Elmasry that sorting on keys is a good substitute for the Fibonacci heap's sorting on subtree sizes, and uses the pairing heap's implementation of \opDc, which is the simplest.

\subsection{Potential}

The potential method is used to bound the amortized runtimes of the operations of the sort heap.

\subsubsection{Preliminaries \checked}

Let $S(x)$ be the number of nodes in the induced subtree of $x$ and its right siblings. Let $L(x)$ be number of nodes in the induced subtree of $x$, and let $R(x)$ be the number of nodes in the induced subtrees of the right siblings of $x$. These definitions are made so that $S(x)=L(x)+R(x)$. 

We call a node \emph{right heavy} if $R(x) \geq \frac{2}{3}S(x)$ and \emph{left heavy} if $L(x) \geq \frac{2}{3}S(x)$. A node that is neither left heavy nor right heavy is called \emph{transitional}; transitional nodes have $L(x)$ and $R(x)$ in the range $(\frac{1}{3}S(x)..\frac{2}{3}S(x))$.

\subsubsection{Potential function \checked}

Given these definitions, we can now define the potential of a node. For ease of presentation, we will first define the \emph{raw potential} of a node $\zeta(x)$, and then define the \emph{potential} of a node, $\varphi(x)$, to be $\zeta(x) c \log \log n$, where $c$ is a sufficiently large constant chosen, where sufficiently large will be defined in the proof of Lemma~\jref{l:em}. This allows us to avoid the clutter of $c \log \log n$ terms when presenting some preliminary lemmas.

 Right heavy nodes have raw potential $1$, and left heavy nodes have raw potential $0$. Transitional nodes have their raw potential linearly vary from $0$ to $1$, specifically, if $x$ is transitional 

$$\zeta(x)=\frac{R(x)-\frac{1}{3}S(x)}{\frac{1}{3}S(x)}.$$

For example, the raw potential is $\frac{1}{2}$ when $R(x)=L(x)=\frac{1}{2}S(x)$.

Observe that the potential of a node varies from~0 to $c\log \log n$. This is a smaller range than that used to analyze other pairing-heap type structures. The original analysis of pairing heaps \cite{DBLP:journals/algorithmica/FredmanSST86} and Elmasry's heaps %
\cite{DBLP:conf/soda/Elmasry09} vary up to $\log n$, while Pettie's analysis of pairing heaps can vary by $2^{\Omega(\log \log n)}$ \cite{DBLP:conf/focs/Pettie05}.

The potential of a heap is the sum of the potentials of the nodes of the heap. 

\subsubsection{Properties of the raw potential}

In this section several facts about the raw potential are presented that form the backbone of the analysis of the structure.

\begin{obs} \checked \jlabel{obs:one}
A single pairing only changes the raw potential of the two nodes directly involved in the pairing; their raw potentials can only change by one.
\end{obs}

\begin{lemma} \jlabel{lem:potSib} \checked
In a heap of size $n$, the sum of raw potentials of $k$ nodes which are mutual siblings is at least $k-\log_\frac{3}{2} n$.
\end{lemma}

\begin{proof}
In order for a node to not have raw potential 1, it must be left heavy or transitional. We will bound the number of such nodes.
Label the mutual siblings which are left heavy or transitional $x_1, x_2, \ldots , x_\ell$ in left-to-right order.
Thus, $R(x_i) \leq \frac{2S(x_i)}{3}$; combining this with the fact that $R(x_i) \geq S(x_{i+1})$ gives $S(x_i) \geq \frac{3}{2}S(x_{i+1})$. Since $S(x_1) \leq n$ and $S(x_k)\geq 1$, this gives $\ell \leq \log_\frac{3}{2} n$, which gives the result.
\end{proof}

\begin{lemma} \jlabel{lem:potPath} \checked
In a heap of size $n$, the sum of the raw potentials of any set $S$ of nodes which are mutual ancestors/descendants (i.e.~a subset of a root-to-leaf path) is at most $\log_\frac{3}{2} n$.
\end{lemma}

\begin{proof}
The proof is similar to that of the previous lemma, with the minor need to account for the asymmetry of a node itself counting in $L(x)$ and not in $R(x)$. Let $x_1, x_2, \ldots, x_\ell$ be those nodes in $S$ which have nonzero raw potential, ordered such that each node is a descendant of those that precede it and an ancestor or those that follow it. Since each node $x_i$ has non-zero raw potential, it must be either right-heavy or transitional, which gives $L(x_i) \leq \frac{2S(x_i)}{3}$. Since $S(x_{i+1}) \leq L(x_i)-1$, we obtain $\frac{3}{2}S(x_i) \geq S(x_{i+1})$. Since $S(x_1) \leq n$ and $S(x_k)\geq 1$, this gives $\ell \leq \log_\frac{3}{2} n$, which gives the result.
\end{proof}

\begin{lemma} \jlabel{lem:remnode}
Removing a node and its subtree from a heap increases the sum of the raw potential of the remaining nodes of the heap by at most $13$.
\end{lemma}

\begin{proof}
The only nodes that can have their raw potential increase are the (former) ancestors of the disconnected node. The siblings to the left of the (former) ancestors of the disconnected node could also have their potential change, but this change will always be negative.
Let  $x_1, x_2, \ldots, x_\ell$ denote these nodes which have their raw potential increase, ordered such that each node is a descendant of those that precede it and an ancestor of those that follow it.

All nodes $x_i$ will have $L(x_i)$ decrease by the size of the removed subtree (call it $M$) while $R(x_i)$ will remain the same. We will use $L(), R(), S()$ and $\varphi()$ to refer to the values before detaching the node, and $L'()=L()-M, R'()=R(), S'()=S()-M$ and $\varphi'()$ to be defined as a function of the state of the structure after detaching the node. The only way these changes can cause an increase in $x_i$'s raw potential is when $x_i$ is left heavy or transitional before the node is disconnected; if it is right heavy, it will remain right heavy with unchanged unit raw potential. For the same reasons as in Lemma~\jref{lem:potPath}, we have $S(x_i) \geq\frac{3}{2} S(x_{i+1})$. This in turn implies that 

$$S(x_i) \geq \left( \frac32 \right)^{\ell-i}S(x_\ell),   $$

 which since $S(x_\ell) \geq M$ gives 
 
 \begin{equation}
 S(x_i) \geq \left( \frac32 \right)^{\ell-i}M . \jlabel{eq:grow}
\end{equation}

Consider the case where $x_i$ changes potential. 
Since it was transitional or left heavy before, and becomes right heavy or transitional after, we know that its increase in raw potential is:

\begin{align}
 \zeta'(x_i) - \zeta(x_i)
& = \overbrace{\min\left(\overbrace{1}^{\substack{\text{right}\\\text{heavy}}},\overbrace{\frac{R'(x_i)-\frac{1}{3}S'(x_i)}{\frac{1}{3}S'(x_i)}}^{\text{transitional}}\right)}^{\text{potential after}} - 
\overbrace{\max \left( \overbrace{0}^{\substack{\text{left}\\ \text{heavy}}},\overbrace{\frac{R(x_i)-\frac{1}{3}S(x_i)}{\frac{1}{3}S(x_i)}}^{\text{transitional}} \right)}^{\text{potential before}} \nonumber \\
\intertext{The use of min and max works because the formula for a transitional node would be at most 0 when applied to a left heavy node and at least 1 when applied to a right heavy node. Removing the min and max can only increase the gain:}
& \leq \frac{R'(x_i)-\frac{1}{3}S'(x_i)}{\frac{1}{3}S'(x_i)} - 
 \frac{R(x_i)-\frac{1}{3}S(x_i)}{\frac{1}{3}S(x_i)}  \nonumber \\
\intertext{Using the fact that $S'(x_i)=S(x_i)-M$, and $R(x_i)=R'(x_i)$: }
 &= \frac{R(x_i)-\frac{1}{3}(S(x_i)-M)}{\frac{1}{3}(S(x_i)-M)} - \frac{R(x_i)-\frac{1}{3}S(x_i)}{\frac{1}{3}S(x_i)}\nonumber   \nonumber \\
 &= \frac{3R(x_i)M}{S(x_i)(S(x_i)-M)} \nonumber \\
\intertext{Since $\frac{R(x_i)}{S(x_i)}<1$}
&<\frac{3M}{S(x_i)-M} \nonumber \\
\intertext{Using \jeqref{eq:grow}:}
&\leq \frac{3M}{\left( \frac32 \right)^{\ell-i}M-M} \nonumber \\
&= \frac{3}{\left( \frac32 \right)^{\ell-i}-1} \jlabel{eq:cutbound}
\end{align}

Thus we can bound the sum of the raw potential increase of the nodes $x_1, x_2, \ldots , x_{\ell-1}$ using equation~\jeqref{eq:cutbound} as follows:

\begin{equation}
\sum_{i=1}^{\ell-1} (\zeta'(x_i)-\zeta(x_i)) \leq
 \sum_{i=1}^{\ell-1} \frac{3}{\left( \frac32 \right)^{\ell-i}-1} < 12 \jlabel{eq:sumcutbound}
 \nonumber
 \end{equation}

Using Mathematica, a closed form solution to the limit of this sum was obtained that has an upper bound of 12. The node $x_\ell$ was excluded from the above sum, as this would cause a divide-by-zero. However, any single node can only change by 1, so including $x_\ell$ will give an upper bound of 13.

As these are the only nodes that can change potential, this completes the lemma.

\end{proof}

\subsubsection{Potential impact resulting from the growth of the structure} \checked

One fact will be of use, is to show that as $n$ changes this does not have a large effect on the potential of the nodes, even though all nodes are multiplied by $\log \log n$.

\begin{fact} \jlabel{f:lln}
If $n \geq 3$, then $n \log \log n - n \log \log (n-1) \leq 2$
\end{fact}

\subsection{Analysis}

\subsubsection{\opIns(x)}

\begin{lemma}
The amortized cost of \opIns\ in a sort heap is $O(\log \log n)$.
\end{lemma}

\begin{proof}
The actual cost is $1$ as a single pairing is performed. View the change of potential in two steps: the first is adjusting the $\log \log n$ multiplier of the potential of all the nodes caused by the incrementing of $n$. Fact~\jref{f:lln} tells us that this increases the potential by at most $O(1)$. Then look at the potential of the new inserted node; this is $O(\log \log n)$ since the raw potential of a node is at most~1 by observation~\jref{obs:one}. The only other change possible is in the potential of the former root of the heap, again this is trivially bounded by $O(\log \log n)$ since this is the maximum attainable potential. Thus combining the actual cost plus the change in potential gives the $O(\log \log n)$ amortized cost.
\end{proof}

\subsubsection{\opDc$(p,\delta x)$}

\begin{lemma}
The amortized cost of \opDc\ in a sort heap is $O(\log \log n)$ \checked
\end{lemma}

\begin{proof}
The actual cost is $1$. We analyze the potential in two steps. Let $y$ be the node pointed to by~$p$. The first step is to detach $y$ from its parent and decrease its key. 
This could cause changes in the potential in the (former) ancestors of $y$. Using Lemma~\jref{lem:remnode}, the increase in potential is bounded to be at most $12 \log \log n$. The second step is to pair $y$ with the root of the tree. This can change the potentials of only $y$ and the root of the tree, giving an easy bound on the potential gain of $2 \log \log n$. Thus combining the actual cost plus the gain in potential gives an amortized cost of at most $1+ 12 \log \log n = O(\log \log n)$.
\end{proof}

\subsubsection{\opEm} \jlabel{s:shem}

\begin{lemma}
The amortized cost of \opEm\ in a sort heap is $O(\log n \log \log n)$. \jlabel{l:em}
\end{lemma}

\begin{proof} Let $\ell=\left\lceil \frac{n}{\consd \log n} \right\rceil$ be the number of blocks which are sorted as part of the \opEm ; they all have size exactly $\consd \log n$ except for possibly the last one.
The actual cost is therefore $O(\ell \log n \log \log n)$, as it costs $O(  \log n \log \log n)$ to sort each block, and the sorting dominates the actual cost.
The key observation is that all nodes in a block were mutual siblings, and because they are being paired in sorted order, they form a vertical ancestor/descendent chain when the pairings \opEm\ is complete.
Thus, by Lemma~\jref{lem:potSib}, in each full block the potentials before were at least 
$c(k-\log_\frac32 n)\log \log n$, since $k=\consd \log n$ this gives $\consd c  \log n \log \log n - c\log_\frac32 n \log \log n$; for the non-full block we simply bound the potentials as being at least 0. By Lemma~\jref{lem:potPath}, in each block the potential after sorting and pairing is at most $c\log_\frac32 n \log \log n$. The only change in potential not yet considered the the removal of the minimum element, which is the root, after all the pairings are complete; this causes a loss of potential and thus can be ingnored. Thus the total amortized cost is given by the sum of the actual costs and change in potential; this is at most:

\begin{align}
& \overbrace{O( \ell \log n \log \log n)}^{\text{actual cost}} +\overbrace{ \ell c \log_\frac32  n \log \log n}^{\text{new potential}}- \overbrace{(\ell-1)( \consd c \log n \log \log n - c\log_\frac32  n \log \log n)}^{\text{old potential}}
\nonumber
\\
& =O(\ell  \log n \log \log n)- \overbrace{\left( \consd -\frac{2}{\log \frac32} \right)}^{\approx 0.58}  \ell c  \log n \log \log n  + 
\overbrace{\left(\consd-\frac{1}{\log \frac32} \right)}^{\approx 2.29}  c\log n \log \log n
\intertext{As we noted in the definition of the potential, the choice of the constant $c$ has been deferred until this point. The constant $c$ is chosen to be large enough so that the 
second term is sufficiently large so as to cancel the big-O expression, thus giving the amortized cost as:}
& \leq \left(\consd-\frac{1}{\log \frac32} \right) c\log_\frac32 n \log \log n= O(\log n \log \log n).
\nonumber
\end{align}

\end{proof}
\end{fullonly}

\section{Comments} \jlabel{s:comments}

Both the proof presented here and the proof of Fredman share many similar aspects. Fredman's restriction that comparisons can not be performed without a pairing allowed him to observe a property he called \emph{consistency}. This is the notion that he did not have to worry about the key values themselves, since their ranks were a perfect proxy for key values. Much of our effort has been spent to get a result without the consistency property. To do this, we have made some changes to the rank function compared to Fredman so that the rank of a node changes more slowly; we attempt to keep a node's rank in sync with its key value, but this is not completely possible. However, we show that through the use of the violation list and additional \opDc\ operations to re-key items whose rank and key value no longer correlate, something like consistency can be managed. The ideas of evolution and having the adversary maintain sets of sequences are new.

\begin{fullonly}
Looking forward, there are still unanswered questions and loose ends for possible future work:

\shortfull{}{\begin{itemize\shortfull{*}{}}}

\shortfull{}{\item} Our lower bound of $\Omega \left( \frac{\log \log n}{\log \log \log n} \right)$ for \opDc, differs by a $\log \log \log n$ factor from the best known pure heaps, and also from Fredman's lower bound. Can our low bound can be improved to remove the triple log, or is there a heap in our model with $o(\log \log n)$ amortized \opDc? Such a heap would necessarily not be in Fredman's model.

\shortfull{}{\item} In our definition of the pure-pointer model of heaps, we do not allow the algorithm to detach a node from its parent unless a \opDc\ is performed on it or if its parent is being removed as part of a \opEm\ operation. The lower bound should be able to be extended to allow such operations, which would make our lower bound apply to Elmasry's variants of heaps.

Similarly, our lower bound does not apply to Fibonacci heaps largely because of the bucket sorting used; this is intentional. However, we described the Pointer Fibonacci heap which made the Fibonacci heap into one with only heap pointers.
However, the Pointer Fibonacci heap remains outside of our pure heap model because the implementation of \opDc\ involves more than just cutting the node from its parent; there is a process called \emph{cascading cuts} whereby, in some cases, ancestors of the node being \opDc ed also are cut from their parents (recall from the introduction, in contrast, rank-pairing heaps are a simplification of Fibonacci heaps that do not use cascading cuts, and thus our lower bound applies to the natural pointer model variant of them). This technicality caused by the cascading cuts places them outside of our model. Fredman's lower bound also would not be able to deal with the cascading cuts, but his sequence of operations used in the lower bound only performs \opDc\ operations on roots and children of the the roots, which prevents Fibonacci heaps from ever invoking the cascading cut. Although our \opDc\ operations performed as part of the permutation evolution also are only performed on roots or children of roots, the \opDc\ operations which are performed on marked nodes as part of processing the violation list could cause cascading cuts to be performed.
We believe that our bound could be extended to allow the algorithm to unlink arbitrary nodes from their parents to cover situations like cascading cuts. 

\shortfull{}{\item} Our sort heap has a $O(\log n \log \log n)$-time \opEm\ and $O(\log \log n)$ \opDc. Is there a pure heap model heap with no augmented data with $O(\log  n)$ amortized time \opEm\ and $O(\log \log n)$ \opDc? We still do not know whether or not pairing heaps are such a heap.
\end{fullonly}

\shortfull{}{\end{itemize\shortfull{*}{}}}

\begin{fullonly}
\section{Acknowledgments}

The idea behind this work came to me while listening to Robert E.~Tarjan give a talk entitled \emph{Rank-Pairing Heaps} 
at the 4th Bertinoro Workshop on Algorithms and Data Structures (ADS) in 2009. I would like to thank the organizers of the workshop, Andrew V. Goldberg,
Giuseppe F. Italiano,  
Valerie King, and
Robert E. Tarjan for inviting me and creating an incredible forum for data structures research.
I would also recognize the contribution of the late Mihai P\v{a}tra\c{s}cu for some conversations we had at ADS 2009, sitting in the shadow of the Colonna delle Anelle, about the possibility of a bound of the type presented here.
Mark Yagnatisky and several anonymous reviewers read earlier versions this paper and pointed out a number of minor bugs. 
Finally I would like to thank my doctoral advisor, Michael L.~Fredman. He had me read his then-new lower bound for my qualifying exam in 1999; without intimate knowledge and understanding of that lower bound, this one would not be possible.
\end{fullonly}

%% file: paperArxiv2.bbl
\begin{thebibliography}{BCFL12}

\bibitem[BCFL12]{DBLP:conf/icalp/BoseCFL12}
Prosenjit Bose, S{\'e}bastien Collette, Rolf Fagerberg, and Stefan Langerman.
\newblock {De-amortizing Binary Search Trees}.
\newblock In {\em {Automata, Languages, and Programming - 39th International
  Colloquium, ICALP 2012, Warwick, UK, July 9-13, 2012, Proceedings, Part I}},
  pages 121--132. {Springer}, 2012.

\bibitem[CLRS09]{clrs}
Thomas~H. Cormen, Charles~E. Leiserson, Ronald~L. Rivest, and Clifford Stein.
\newblock {\em Introduction to Algorithms (3. ed.)}.
\newblock MIT Press, 2009.

\bibitem[Dij59]{springerlink:10.1007/BF01386390}
E.~W. Dijkstra.
\newblock A note on two problems in connexion with graphs.
\newblock {\em Numerische Mathematik}, 1:269--271, 1959.
\newblock 10.1007/BF01386390.

\bibitem[Elm09]{DBLP:conf/soda/Elmasry09}
Amr Elmasry.
\newblock {Pairing heaps with {\it O}(log log {\it n}) decrease cost}.
\newblock In {\em {Proceedings of the Twentieth Annual ACM-SIAM Symposium on
  Discrete Algorithms, SODA 2009, New York, NY, USA, January 4-6, 2009}}, pages
  471--476. {SIAM}, 2009.

\bibitem[Elm10a]{DBLP:conf/esa/Elmasry10}
Amr Elmasry.
\newblock {Pairing Heaps with Costless Meld}.
\newblock In {\em {Algorithms - ESA 2010, 18th Annual European Symposium,
  Liverpool, UK, September 6-8, 2010. Proceedings, Part II}}, pages 183--193.
  {Springer}, 2010.

\bibitem[Elm10b]{DBLP:journals/dmaa/Elmasry10}
Amr Elmasry.
\newblock {The Violation Heap: a Relaxed Fibonacci-like Heap}.
\newblock {\em Discrete Math., Alg. and Appl.}, 2(4):493--504, 2010.

\bibitem[Fre99a]{DBLP:conf/wae/Fredman99}
Michael~L. Fredman.
\newblock {A Priority Queue Transform}.
\newblock In {\em {Algorithm Engineering, 3rd International Workshop, WAE '99,
  London, UK, July 19-21, 1999, Proceedings}}, pages 244--258. {Springer},
  1999.

\bibitem[Fre99b]{DBLP:journals/jacm/Fredman99}
Michael~L. Fredman.
\newblock {On the Efficiency of Pairing Heaps and Related Data Structures}.
\newblock {\em J. ACM}, 46(4):473--501, 1999.

\bibitem[FSST86]{DBLP:journals/algorithmica/FredmanSST86}
Michael~L. Fredman, Robert Sedgewick, Daniel~Dominic Sleator, and Robert~Endre
  Tarjan.
\newblock {The Pairing Heap: A New Form of Self-Adjusting Heap}.
\newblock {\em Algorithmica}, 1(1):111--129, 1986.

\bibitem[FT87]{DBLP:journals/jacm/FredmanT87}
Michael~L. Fredman and Robert~Endre Tarjan.
\newblock {Fibonacci heaps and their uses in improved network optimization
  algorithms}.
\newblock {\em J. ACM}, 34(3):596--615, 1987.

\bibitem[HST11]{DBLP:journals/siamcomp/HaeuplerST11}
Bernhard Haeupler, Siddhartha Sen, and Robert~Endre Tarjan.
\newblock {Rank-Pairing Heaps}.
\newblock {\em SIAM J. Comput.}, 40(6):1463--1485, 2011.

\bibitem[Iac00]{DBLP:conf/swat/Iacono00}
John Iacono.
\newblock {Improved Upper Bounds for Pairing Heaps}.
\newblock In {\em {Algorithm Theory - SWAT 2000, 7th Scandinavian Workshop on
  Algorithm Theory, Bergen, Norway, July 5-7, 2000, Proceedings}}, pages
  32--45. {Springer}, 2000.

\bibitem[Iac11]{DBLP:journals/corr/abs-1110-4428}
John Iacono.
\newblock {Improved Upper Bounds for Pairing Heaps}.
\newblock {\em CoRR}, abs/1110.4428, 2011.

\bibitem[KT08]{DBLP:journals/talg/KaplanT08}
Haim Kaplan and Robert~Endre Tarjan.
\newblock {Thin heaps, thick heaps}.
\newblock {\em ACM Transactions on Algorithms}, 4(1), 2008.

\bibitem[Pet05]{DBLP:conf/focs/Pettie05}
Seth Pettie.
\newblock {Towards a Final Analysis of Pairing Heaps}.
\newblock In {\em {46th Annual IEEE Symposium on Foundations of Computer
  Science (FOCS 2005), 23-25 October 2005, Pittsburgh, PA, USA, Proceedings}},
  pages 174--183. {IEEE Computer Society}, 2005.

\bibitem[ST85]{DBLP:journals/jacm/SleatorT85}
Daniel~Dominic Sleator and Robert~Endre Tarjan.
\newblock {Self-Adjusting Binary Search Trees}.
\newblock {\em J. ACM}, 32(3):652--686, 1985.

\bibitem[ST86]{DBLP:journals/siamcomp/SleatorT86}
Daniel~Dominic Sleator and Robert~Endre Tarjan.
\newblock {Self-Adjusting Heaps}.
\newblock {\em SIAM J. Comput.}, 15(1):52--69, 1986.

\bibitem[SV87]{DBLP:journals/cacm/StaskoV87}
John~T. Stasko and Jeffrey~Scott Vitter.
\newblock {Pairing Heaps: Experiments and Analysis}.
\newblock {\em Commun. ACM}, 30(3):234--249, 1987.

\bibitem[Wil89]{DBLP:journals/siamcomp/Wilber89}
Robert~E. Wilber.
\newblock {Lower Bounds for Accessing Binary Search Trees with Rotations}.
\newblock {\em SIAM J. Comput.}, 18(1):56--67, 1989.

\end{thebibliography}
